\documentclass[twocolumn,aps,pra,showpacs,superscriptaddress,nofootinbib]{revtex4}
\usepackage[usenames,dvipsnames]{color}
\pdfoutput=1
\usepackage{graphicx}

\usepackage{amsmath,amssymb,amsthm}
\usepackage{color}
\usepackage{graphics}
\usepackage[FIGTOPCAP,raggedright,nooneline]{subfigure}
\usepackage{multirow}
\usepackage{mathtools}
\usepackage{array}
\usepackage[british]{babel}

   \theoremstyle{plain}
 \newtheorem{lem}{Lemma}
  
  \newcommand{\e}{\mathrm{e}}
\newcommand{\ot}{\otimes}

\newcommand{\I}{\mathbb{I}} 
 
\newcommand{\D}{\ensuremath{\mathcal{D}}}
\newcommand{\UU}{\ensuremath{\mathcal{U}}} 
\renewcommand{\H}{\mathcal{H}} 
\newcommand{\V}{\mathcal{V}} 

\newcommand{\PPP}{\mathbf{P}}
\newcommand{\MMM}{\mathbf{M}}

\DeclareMathOperator{\QFI}{\mathnormal{F_Q}\!} 
\DeclareMathOperator{\FI}{\mathnormal{F}_\mathrm{cl}} 
\DeclareMathOperator{\tr}{tr} 
\newcommand{\ii}{\mathrm{i}} 
\DeclareMathOperator{\param}{\varphi} 
\DeclareMathOperator{\state}{\varrho} 

\newcolumntype{M}[1]{>{\centering\arraybackslash}m{#1}}
\newcolumntype{N}{@{}m{0pt}@{}}

\makeatletter
\@ifundefined{textcolor}{}
{%
 \definecolor{BLACK}{gray}{0}
 \definecolor{WHITE}{gray}{1}
 \definecolor{RED}{rgb}{1,0,0}
 \definecolor{GREEN}{rgb}{0,1,0}
 \definecolor{BLUE}{rgb}{0,0,1}
 \definecolor{CYAN}{cmyk}{1,0,0,0}
 \definecolor{MAGENTA}{cmyk}{0,1,0,0}
 \definecolor{YELLOW}{cmyk}{0,0,1,0}
}

\makeatother

\usepackage{babel}

\begin{document}
\global\long\global\long\global\long\def\bra#1{\mbox{\ensuremath{\langle#1|}}}
\global\long\global\long\global\long\def\ket#1{\mbox{\ensuremath{|#1\rangle}}}
\global\long\global\long\global\long\def\bk#1#2{\mbox{\ensuremath{\ensuremath{\langle#1|#2\rangle}}}}
\global\long\global\long\global\long\def\kb#1#2{\mbox{\ensuremath{\ensuremath{\ensuremath{|#1\rangle\!\langle#2|}}}}}

\global\long\global\long\global\long\def\SET#1#2{\mbox{\ensuremath{\ensuremath{\left\lbrace\left. #1\ \right|\ #2 \right\rbrace }}}}

\newcommand{\id}{\openone}
\newcommand{\NN}{\ensuremath{\mathcal{N}}}
\newcommand{\PP}{\ensuremath{\mathcal{P}}}
\newcommand{\GHZ}{\ensuremath{\mathrm{GHZ}}}

\newcommand{\ODF}{\ensuremath{\mathrm{ODF}}}
\newcommand{\DFS}{\ensuremath{\mathrm{DFS}}}
\newcommand{\braket}[1]{\ensuremath{\langle #1\rangle}}

\title{Estimation of gradients in quantum metrology}
\author{Sanah Altenburg}
\thanks{These two authors contributed equally to this work.}
\affiliation{Naturwissenschaftlich-Technische Fakult\"at,
Universit\"at Siegen,
Walter-Flex-Str.~3,
D-57072 Siegen}

\author{Micha\l\ Oszmaniec}
\thanks{These two authors contributed equally to this work.}
\affiliation{ICFO-Institut de Ciencies Fotoniques, The Barcelona Institute of Science and Technology, 08860 Castelldefels (Barcelona), Spain}
\affiliation{National Quantum Information Centre of Gda\'nsk, 81-824 Sopot, Poland}
\affiliation{Faculty of Mathematics, Physics and Informatics, 
	Institute of Theoretical Physics and Astrophysics, University of
	Gda\'nsk, 80-952 Gdañsk, Poland}

\author{Sabine W\"olk}
\author{Otfried G\"uhne}
\affiliation{Naturwissenschaftlich-Technische Fakult\"at,
Universit\"at Siegen,
Walter-Flex-Str.~3,
D-57072 Siegen}

\date{\today}


\begin{abstract}
We develop a general theory to estimate magnetic field gradients in
quantum metrology.  We consider a system of $N$ particles distributed 
on a line whose internal degrees of freedom interact with a magnetic 
field.  Usually gradient estimation is based on precise measurements 
of the magnetic field at two different locations, performed with two 
independent groups of particles. This approach, however, is sensitive 
to fluctuations of the off-set field determining the level-splitting 
of the particles and results in collective dephasing. In this work we use the framework of quantum metrology to assess the maximal accuracy for gradient estimation.  For arbitrary positioning of particles, we identify optimal entangled and separable states  allowing the estimation of gradients with the maximal accuracy, quantified by 
the quantum Fisher information.   We also analyze the performance of states from the decoherence-free subspace (DFS),  which are insensitive to the fluctuations of the magnetic offset field. We find that these states allow  to measure 
 a gradient directly, without the necessity of estimating the magnetic offset field.  Moreover, we show that  DFS states attain a precision for gradient estimation  comparable to the optimal entangled states. Finally, for the above classes of states  we find simple and feasible measurements saturating the quantum Cram\'{e}r-Rao bound.
\end{abstract}

\maketitle
\section{Introduction}

Quantum metrology holds the promise to enhance the measurement of physical 
quantities with the help of quantum effects \cite{Giovannetti2006, Toth2014}. 
In practice, ideas from quantum metrology may improve gravitational wave 
detectors \cite{schnabel}, imaging in biology \cite{biology} or sensors
for protein molecules \cite{jelezko}. In a typical metrological scenario 
one aims to estimate a certain phase $\varphi$, e.g., generated by a 
magnetic field, with quantum probe systems. By using entanglement between 
the probes, the uncertainty $\Delta^2\tilde{\varphi}$ of the estimate  can be
reduced \cite{Giovannetti2006}. In this way, quantum metrology offers 
an advantage in theory, but for practical implementations noise and 
decoherence have to be taken into account. Here, it has been shown 
that noise has often a negative effect and the improved scaling gets 
lost \cite{Escher2011, Demkowicz-Dobrzanski2012}. Nevertheless, concepts
such as differential metrology, where some probe systems are used to monitor
the noise, can be used to maintain a quantum advantage \cite{Landini2014, 
Altenburg2016}.

A different problem is the estimation of the gradient of a spatially 
distributed magnetic field \cite{othergradientpapers, inigo, Schmidt-Kaler2012a}. Of course, 
one may just measure the field at different positions \cite{Snadden1998} or move a single probe through the field \cite{Walther2011,Keenan2012}, and then compute 
the gradient. But these are not necessarily the optimal strategies, especially 
in cases where one aims to measure small fluctuations of a large offset 
field. Then, a detection of magnetic fields with high precision and spatial resolution is often not possible \cite{ Wildermuth2005}. 
Furthermore, techniques to measure spatial varying fields by probes with well known position can be reversed in order to measure the  spatial distributions of probes \cite{MRIpapers} or the spatial distribution of entanglement \cite{Woelk2016a}
 by a well known spatial field distribution. For example in magnetic resonance imaging (MRI) the spatial resolution of such an image depends on the strength of the applied magnetic field gradient and the calibration of this gradient.
 The larger the magnetic field gradient, the better the resolution. 
 However, practically the resolution is limited by the patients, e.g., for patients with medical implants. An old cardiac pacemaker or a cochlea implant would make the application of a large magnetic field gradient and therefore a high spatial resolution MRI impossible. 
A precise calibration of the applied spatial field distribution is necessary.
 Here, quantum metrology could offer a solution for a precise calibration. With the here presented findings it will be possible to calibrate a gradient for MRI with high precision. 

 In this paper we discuss the estimation of magnetic gradients
using the language of quantum metrology. We consider $N$ particles distributed
in an arbitrary but fixed manner along a line, and ask in which quantum
state they have to be prepared and which measurements have to be carried
out in order to estimate the magnetic field gradient with the optimal
precision. We also consider the case of collective dephasing noise, as it occurs in 
realistic set-ups with trapped ions \cite{Monz2011} or neutral atoms in optical microtraps. We arrive at a general
scheme with optimal states and measurements, depending on the knowledge
of the offset field, or on the presence of noise. 

This paper is organized as follows: In Section II we explain basic facts 
about quantum metrology and the Fisher information, being the central 
figure of merit in estimation scenarios. In Section III we introduce the
scenario of gradient estimation. Section IV deals with the case that the
offset field $B_0$ at a certain position is known and the gradient should
be estimated. In this part we also consider the effect of collective phase noise on the performance of gradient estimation. Section V considers the situation, where the offset field
$B_0$ is not known. Section VI discusses shortly the measurement of more
general notions than the gradient of the field. Finally, we conclude and 
discuss the optimal strategies.

\section{Quantum metrology}\label{sec:estimation_theory}
\begin{figure*}
\begin{center}
\subfigure[ ]{\includegraphics[width=0.4\textwidth]{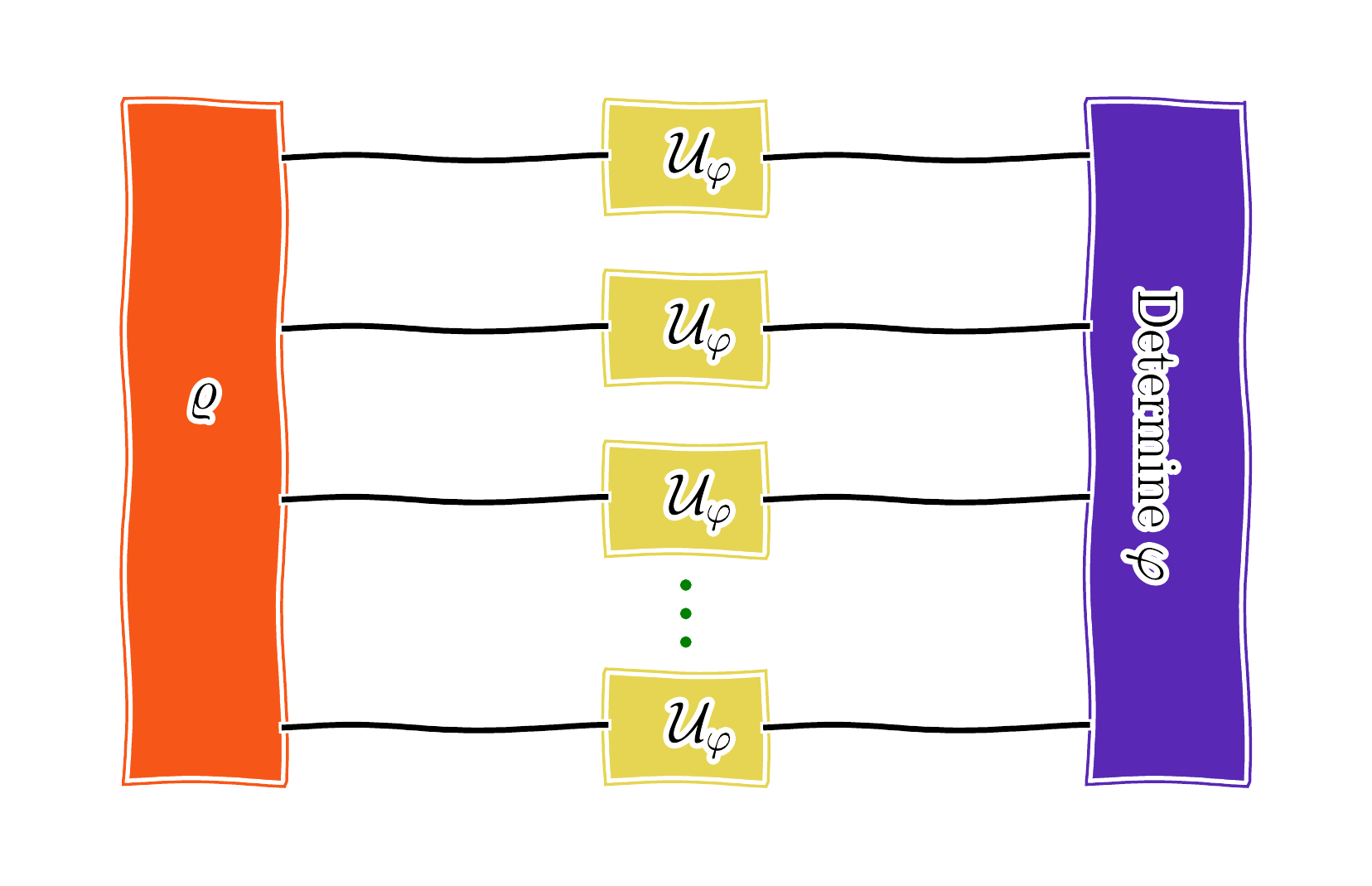}}
\subfigure[ ]{\includegraphics[width=0.4\textwidth]{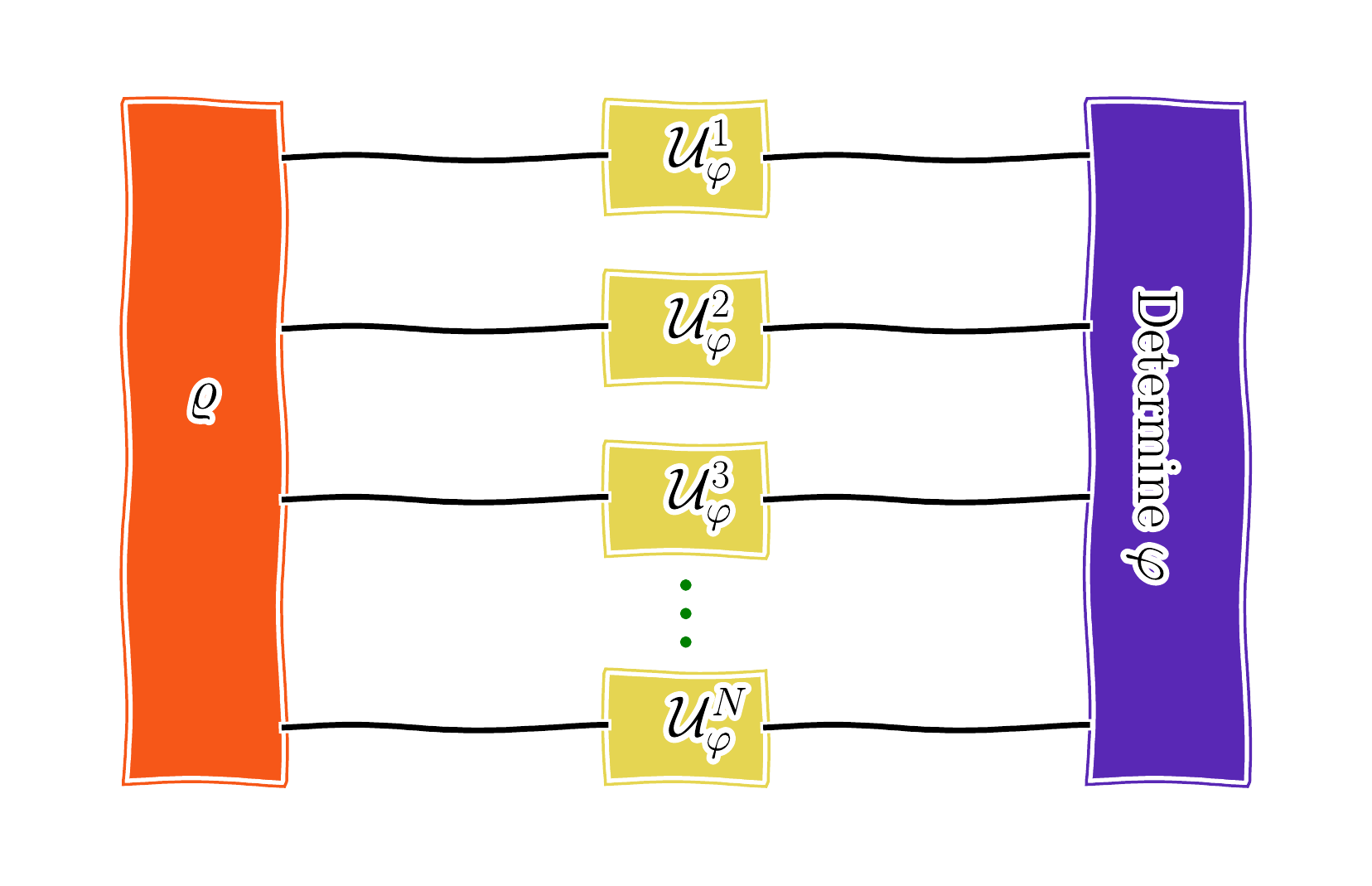}}
\end{center}
\caption{Comparison of two metrological scenarios (in the absence of experimental noise). A quantum device encodes the parameter  $\param$ by acting on the multiparticle initial state $\state$. The parameter is estimated by multiple measurements of the output state. \textbf{(a)}: The parameter $\param$ is encoded by the application of unitaries  $\UU_{\param}$ acting in parallel on every particle. \textbf{(b)}   The parameter $\param$ is encoded by the application of unitaries $\UU^{(i)}_{\param}$ applied in parallel, but  acting differently on each particle. This scenario is relevant for the gradient estimation.}
\label{fig:shemes}
\end{figure*}

We first review the basics of quantum metrology \cite{Giovannetti2006} and introduce the main technical tools and notation that will be used in our work. Then we briefly present the canonical phase estimation scheme and compare it with gradient estimation that is studied in this work.

The task in a typical quantum metrology scheme is to determine an unknown parameter $\param$ which is encoded in a quantum state $\state_{\param}$ by a quantum channel described via the map $\Lambda_{\param}$.
After passing through the quantum channel, the state is subsequently measured and the whole process repeated $\nu$ times to gather the sufficient statistics.

Let $p_{j}(\param)$ be the probability for the measurement outcome $j$, given that the initial state was $\state$ and the unknown parameter was $\param$. Then a  result in classical statistics states that the variance $\Delta^2\tilde\param$ of any unbiased and consistent estimator $\tilde\param$ of $\param$ is 
lower-bounded by the \emph{Cram\'{e}r-Rao Bound} (CRB) \cite{ClassStat}: 
\begin{equation} \label{eq:CRB}
\Delta^2\tilde\param
\geq
\frac{1}{\nu\FI\left(\left\{p_{j}(\param)\right\}\right)},
\end{equation}
where
\begin{equation} \label{eq:FI}
\FI\left(\left\{p_{j}(\param)\right\}\right) \coloneqq \sum_j \frac{\left[ \partial_\varphi p_j(\varphi) \right]^2}{p_j(\varphi)} 
\end{equation}
is the \emph{classical Fisher information} (FI). A single estimator saturating Eq.~\eqref{eq:CRB} may not always exist for the whole parameter range. When estimating small fluctuations of a parameter around a given value, the CRB in Eq.~\eqref{eq:CRB} is guaranteed to be tight in the limit of large number of repetitions  $\nu$  \cite{ClassStat}.

If the probabilities $p_j(\param)$ come from a quantum mechanical experiment, the classical Fisher information (FI) $\FI$ depends on the initial state $\varrho$, the map $\Lambda_{\param}$ encoding the phase
\begin{equation}
\Lambda_{\param}: \state \xrightarrow{} \varrho_\varphi\coloneqq\Lambda_{\varphi}(\state)
\end{equation}
 and the performed measurement. In quantum mechanics the measurement process is described by a Positive Operator-Valued Measure (POVM), i.e. a collection $\MMM=\{M_j\}$ of positive semi definite operators satisfying the normalization condition $\sum_j M_j=\id$. The probability of measuring the outcome $j$ on the state $\state_{\param}$ is given by
\begin{equation}\label{eq:ClassicalFisher}
p_j(\param)=\mathrm{Tr}\left[M_j \Lambda_{\varphi}(\varrho)\right].
\end{equation}
In the following we will denote the classical Fisher information for the measurement statistics obtained from $\state_{\param}$ by the POVM $\MMM$ by $\FI(\rho_{\param},\MMM)$. 
The optimization of $\FI(\rho_{\param},\MMM)$ over all possible POVMs is called \emph{Quantum Fisher Information} $\QFI[\rho, \Lambda_{\param}]$ (QFI) \cite{Caves1994}.
The QFI depends solely on the quantum state $\state$ and $\Lambda_{\param}$, whereas the FI depends on the state $\state$, $\Lambda_{\param}$ and the POVM  $\MMM$. The QFI operationally quantifies the metrological usefulness of the initial state $\state$ under the map $\Lambda_{\param}$ for the estimation of $\param$. The precision limitations for estimating $\param$ is usually put in the from of the \emph{quantum Cram\'{e}r-Rao Bound} 
\begin{equation}\label{eq:Cramér_rao2}
\Delta^2 \tilde{\param} \ge \frac{1}{\nu \QFI[\state, \Lambda_{\varphi}]}\ .
\end{equation}
The QFI $\QFI$ can be computed explicitly via the formula \cite{Toth2014,RDD2015}, 
\begin{equation} \label{eq:QFI}
\QFI[\state, \Lambda_{\varphi}]=2\sum_{\alpha,\beta:\lambda_\alpha +\lambda_\beta \neq 0} \frac{|\bk{\alpha|\partial_\varphi \Lambda_{\param}(\state)}{\beta}|^2}{\lambda_\alpha + \lambda_\beta}\,,
\end{equation}
where  $\{\lambda_\alpha\}$ are the eigenvalues and  $\{\ket{\alpha}\}$ the eigenvectors of $\Lambda_{\param}(\state)$. If the parameter $\param$ is encoded via an unitary evolution, i.e. when $\Lambda_{\param}(\state)=U_{\param} \state U^{\dagger}_{\param}$, where $U_{\param}=\exp(-\ii \param  H)$   for some Hermitian operator $H$, then the QFI depends only on the initial state $\state$ and the operator $H$ and will be denoted by $\QFI\left[\state,H\right]$. The QFI for pure states $\psi\coloneqq \kb{\psi}{\psi}$  in unitary time evolutions is related to the variance of the operator $H$,
\begin{equation}
\QFI\left[\psi,H\right]=4 \Delta^{2}_\psi H \coloneqq 4\left[\tr\left(H^2 \psi \right) - \tr\left(H \psi \right)^2   \right]\ . 
\end{equation}
 Let us also recall that $\QFI\left[\state,H\right]$ is a convex function of $\state$. For this reason the maximal value of QFI is  always attained for pure states. In fact, for a fixed Hermitian operator $H$, the maximal $\QFI$ can be computed explicitly by \cite{Giovannetti2006}
\begin{equation}\label{eq:maxfisch}
\max_{\state\in\D(\H)} \QFI\left[\state,H\right]=\left(\lambda_{\max}-\lambda_{\min}\right)^{2}\ ,
\end{equation}
where $\lambda_{\max}$ and $\lambda_{\min}$ are the
maximal and   minimal eigenvalues of $H$ respectively, and $\D(\H)$ denotes the set of (mixed and pure) quantum states supported on the Hilbert space $\H$. The pure
state for which the QFI attains Eq.~\eqref{eq:maxfisch} is given by 
\begin{equation}\label{eq:optSTATE}
\ket{\psi_{\mathrm{opt}}}=\frac{1}{\sqrt{2}}\left(\ket{\max}+\ket{\min}\right)\ ,
\end{equation}
where $\ket{\max}$ and $\ket{\min}$ are eigenvectors of $H$ corresponding
to eigenvalues $\lambda_{\max}$ and respectively $\lambda_{\min}$.

In the experimental context it is common to infer the value of the parameter solely from the expectation value $\langle M \rangle_{\param}\coloneqq \tr (\state_{\param}M)$ of some observable $M$. This is done by using the Taylor expansion 
\begin{equation}\label{eq:approximation}
\tilde{\param}_M - \param_0 \approx  \frac{\langle M \rangle_{\param} - \langle M \rangle_{\param_0}}{ \left. \partial_{\param} \langle M \rangle\right|_{ {\param}_0} }\ ,
\end{equation}
to construct the estimator $\tilde{\param}_M$ \footnote{More specifically, in order to construct the estimator $\tilde{\param}_M$, one has to assume that the statistical fluctuations of the phase $\param$ around the \emph{known} value $\param_0$ are small (so that Eq.~\eqref{eq:approximation} makes sense) and that the expectation value \unexpanded{$\langle M\rangle_{\param_0}$} is known.} 
 of the value of $\param$. This strategy is in general only optimal for a specific choice \footnote{Note, however, that locally (in the neighborhood of the specific value $\param_0$) the precision attainable by this method saturates the quantum Cram\'{e}r-Rao bound given in Eq.~\eqref{eq:Cramér_rao2}, for the suitable choice of the observable. The optimal observable in general depends on the value of the phase. In particular, it is known that the error propagation formula in Eq.~\eqref{eq:error_propagation_formula} for the so-called symmetric logarithmic derivative \cite{Toth2014} saturates Eq.~\eqref{eq:Cramér_rao2}. However, in general there is no guarantee that symmetric logarithmic derivative is an observable easily accessible in an experiment.}
  of the operator $M$.   The precision of this estimator  $\Delta^2 \tilde{\param}_{ M}$, after the experiment is repeated $\nu$ times, is given by
the error-propagation formula \cite{Toth2014}

\begin{equation}\label{eq:error_propagation_formula}
\Delta^2 \tilde{\param}_{M}=\frac{\left(\Delta^2_{\varrho_{\varphi_0}} M\right)^2}{\nu \left[\left. \partial_{\param} \langle M \rangle \right|_{ {\param}_0}\right]^2} \ .
\end{equation}
In what follows, we will drop the number of repetitions $\nu$ in order to simplify the notation and discussion.
\newline
\textbf{Standard metrological scenario:} In the standard scenario a quantum device (e.g., an interferometer) acts on a single particle (photon, atom, etc.) with the Hamiltonian $h_0$ (often taken to be equal to $\frac{1}{2}\sigma_z$).
The device  encodes the unknown parameter $\param$ on the system of $N$ particles by performing the parallel unitary  transformation $U_{\param} = \UU_{\param}^{\ot N}$, where   $\UU_{\param}=\e^{-\ii h_0 \param}$ [see Figure  \ref{fig:shemes}(a)]. This unitary evolution is generated by the global Hamiltonian  $
 H = \sum_{i}^N h_0^{(i)}$, where $h_0^{(i)}$ denotes the Hamiltonian on the $i$th particle. In classical measurement strategies (corresponding to separable input sates),  the particles are only classical correlated and the variance $\Delta^2\tilde{\param}$ for measuring $\varphi$ is limited by the number of particles $N$ via the standard quantum limit (SQL) $\Delta^2\tilde{\param}\propto 1/N$. However, in quantum mechanics we have the freedom to apply the device in parallel to an entangled state of $N$ particles [see Figure \ref{fig:shemes} (a)].  This allows to obtain the accuracy  $\Delta^2\tilde{\param}\propto 1/N^2$, which is usually referred as the Heisenberg limit (HL).
 
 The concepts of SQL and HL are tailored to the schemes where every particle is affected by the same unitary $\UU_{\param}$. As we will see later, in the context of the estimation of gradients of electric or magnetic fields, it is natural to consider again parallel  encoding, but this time allowing single particle unitaries $U^{i}_{\param}$ acting differently on different particles - see   Figure \ref{fig:shemes} (b). The standard HL and the SQL are no longer valid and new bounds in precision have to be derived. This is one of the main aim of the present paper.

At this point, it is important to remark that in the standard metrological scenario, the Heisenberg scaling is typically destroyed by local noise and asymptotically only an enhancement by a constant factor can be achieved \cite{Demkowicz-Dobrzanski2012,Escher2011}. However, in the case of global noise such as collective phase noise the scaling $\Delta^2 \tilde{\param}\propto 1/ N^2$  can be restored, e.g., by differential interferometry \cite{Landini2014, Altenburg2016}. 
In this work we will also discuss the impact of collective phase noise on the performance for gradient estimation.

\section{Set-up for gradient estimation}\label{sec:gradient_estimation_theory}
\begin{figure}
\begin{center}
\includegraphics[width=.45\textwidth]{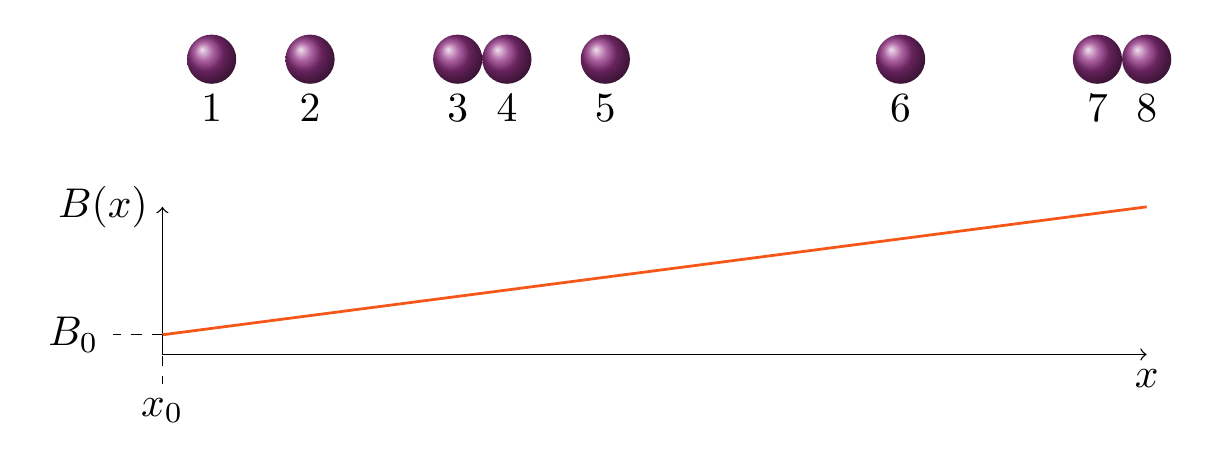}
\end{center}
\caption{A string of particles in a magnetic field with a spatial gradient in the $z$ component along the string. The magnetic field $B(x):=B(x) \vec{e}_z$ acts on each particle, depending on its position. The particles are labeled such that the smallest magnetic field acts on the first and the largest magnetic field on the last particle.
}\label{fig:particles_in_gradient}
\end{figure}
Throughout this paper we consider a string of $N$ particles whose internal, qubit-like, degrees of freedom are coupled to the $z$ component of the spatially-varying magnetic field $B(x):=B(x) \vec{e}_z$ with
\begin{equation}\label{eq:MAGfield}
B(x)=B_0 + (x-x_0)G,
\end{equation}
with $B_0:=B(x_0)$ being the field at position $x_0$, called offset field, and $G$ being the strength of the gradient. Usually in experiments the offset field $B_0$ is set to split the degenerate energetic levels. The direction of the offset field $B_0$ is defined to be the quantization axis that is called $z$ axis. Furthermore, the offset field $B_0$ is strong comparing to fields that point in other directions.  Therefore, we neglect all other components.
We assume without loss of generality a spatial gradient in the $x$ direction. The particles are arranged along the $x$ axis and labeled in such a way that $B(x_i) \le B(x_{i+1})$, where $x_i$ with $i \in \{1,\cdots, N\}$ denotes the position of the $i$th qubit - see Fig.~\ref{fig:particles_in_gradient}.

The magnetic field $B(x)$ depends on the positions $x_i$ of the qubits. In our analysis we will focus on experiments, where these positions can be measured with high precision. This is the case, e.g., in experiments with trapped ions \cite{Monz2011,Leibfried2004} or neutral atoms in optical microtraps \cite{Schlosser2011,Bloch2008}. In both kinds of experiments the position of the qubits can be measured up to $\sim\,$nm, whereas the distance of the qubits scales with $\sim\,\mu$m. Generally, position dependent fields such as magnetic gradients lead to a coupling between the internal and external degrees of freedom \cite{Mintert2001,Woelk2016}. Within this paper, we will assume that this coupling is negligible small, such that the position $x_i$ of the $i$th qubit can be treated classically. This can be always achieved by, e.g., trapping the particles strong enough. 
 
 The Hamiltonian describing the interaction of the internal degrees of freedom with the magnetic field is given by $H=\hbar \gamma H_0 + \hbar \gamma G H_G$, with
\begin{equation}\label{eq:hamiltonian}
H_0 \coloneqq B_0 J_z\ ,\  H_G\coloneqq \frac{1}{2} \sum_{i=1}^N (x_i -x_0)\sigma_z^{(i)}, 
\end{equation}
 where $\sigma_z^{(i)}$ denotes  the Pauli matrix  acting on the $i$th qubit,  $J_z=\frac{1}{2}\sum \sigma_z^{(i)}$ is the collective spin operator, and $\gamma$ is the coupling strength. The map $\Lambda_G$ describing the unitary time evolution due to $H$ for time $t$ is given by $\Lambda_G(\state)=U_G\state U_G^\dagger$, where
\begin{equation}\label{eq:timeevolution}
U_G\coloneqq\mathrm{exp}\left[-\ii  \gamma B_0 t J_z - \ii \gamma G t \sum_{i=1}^N (x_i-x_0) \frac{\sigma_z^{(i)}}{2}\right]. 
\end{equation}

 In the following, we will use tools of quantum metrology to derive limits in precision for a classical and quantum enhanced estimation of the gradient $G$, analogous to SQL and HL known from the standard phase estimation scheme depicted in Fig.~\ref{fig:shemes}(a). 
 The maximal achievable precision of $G$ depends on the knowledge about the magnetic offset field $B_0$. 
In general, an experimenter could measure the offset field with some uncertainty $\Delta^2 \tilde{B}_0$ first and would have some \textit{a priori} knowledge about the offset field 
 before estimating the gradient $G$. 
However, throughout this paper we focus on two extremal situations: full and no \textit{a priori} knowledge about $B_0$. The first scenario is coverd in Section \ref{sec:full_a_priori_knowledge} and allows to study the \glqq{}clean\grqq{} situation, where the only parameter to-be-estimated is the gradient of the field and leads to the ultimate bounds in precision. The second scenario is described in Section \ref{sec:no_a_priori_knowledge} and applies to two interconnected  cases where either (i) an experimenter has no access to the reference frame associated to the rotations generated by the offset field or (ii) the system is strongly affected by the action of collective phase noise. 

As we discuss in Section \ref{sec:estimatin_with_noise} collective phase noise is the main source of noise in setups of trapped ions or trapped atoms in an optical micro-trap.  It leads to an effective erasure of the information about the offset field $B_0$. In Section~\ref{sec:compare_strategies} we will argue that in these experimental scenarios  the precision in gradient estimation does not gain much from the measurements  of the offset field or having a partial knowledge about it. Hence, we  \emph{a fortiori} justify why we focused on the two extreme cases of full and no \emph{a priori} knowledge about $B_0$. Other experimental scenarios  may require a more refined analysis of the problem such as multiparameter estimation \cite{Ragy2016,Proctor2017} and systematically taking into account the lack of knowledge about $B_0$  \cite{OpDeph2013}. 
\newline
Please note that in the main text of the paper we will make two implicit assumptions on the system that we are considering. First, we will assume that $x_i\geq x_0$ for all particles. This affects the precise form of the optimal states and the formula for the maximal QFI. For the sake of simplicity of the presentation we decided to discuss this in detail in Appendix~\ref{app:maxFI}. Second, we will assume that the number of particles $N$ is even. This has only a slight effect on the form of the results for the case of no \emph{a priori} knowledge about $B_0$. The change  for odd $N$ is that the summation range has to be changed from $N/2$ to $\lfloor N/2 \rfloor$. These are discussed in full generality  in  Appendix~\ref{app:presicion_bounds_not_knowing_b}.

\section{Gradient estimation with full \textit{a priori} knowledge about $B_0$}
\label{sec:full_a_priori_knowledge}
 Having the full \textit{a priori} knowledge about the offset field $B_0$ amounts to treating it as a fixed constant. Using the commutation relations 
\begin{equation}
\left[J_z, H_G \right]=\left[U_G,H_G \right]=0\ .
\end{equation}
and Eq.~\eqref{eq:QFI} we obtain the following relation
\begin{equation}\label{eq:simpl QFI}
\QFI\left[\state,\Lambda_G\right]=(\gamma t)^2  \QFI\left[\state,H_G\right] \ .
\end{equation}
In what follows we will reserve the notation
\begin{equation}\label{eq:simplifiedNOT}
F_Q(\state)\coloneqq \QFI\left[\state,\Lambda_G \right]
\end{equation}  
to avoid ambiguity and simplify the notation.  The physical meaning of Eq.~\eqref{eq:simpl QFI} is that the QFI for gradient estimation is reduced to the QFI for the \glqq{}standard Hamiltonian\grqq{} $\QFI\left[\state,H_G\right]$, and that $\QFI\left[\state,\Lambda_G\right]$ does not depend on the value of  the magnetic offset field $B_0$ at $x_0$. However, the QFI does depend on $x_0$, via the dependence of $H_G$ on this parameter - see Eq.~\eqref{eq:hamiltonian}. Notice that the unitary transformation generated by the field has a  product structure $U_G=\otimes_{i=1}^{N}\UU^{i}_G$, where $\UU^{i}_{G}=\exp\{-\ii \gamma t [B_0 + G (x_i -x_0)] \sigma^{(i)}_z /2\}$ is a single qubit unitary. Therefore, the problem of deriving the maximal QFI and the optimal state becomes mathematically equivalent to the case of parallel encoding of the phase given in Figure~\ref{fig:shemes} (b).

The rest of this section is organized as follows. First, in Part A we  identify the bounds in precision for the estimation of the gradient $G$ with separable and entangled states. Then, in Part B we give simple, physically-accessible measurements saturating these bound. Finally, in Part C, we  discuss the influence of collective phase noise on the proposed gradient estimation scheme.

\subsection{Bounds on precision for gradient estimation}
Here, we first derive precision bounds for fixed positions $\lbrace x_i\rbrace$ and identify optimal probe states. Then, we discuss the case of linear spacing of particles. Finally, we identify the optimal positioning of qubits and give the ultimate bounds for gradient estimation.
\newline 
\textbf{Separable states:} Our first result concerns the maximal QFI for estimating $G$ using separable states. We start with the observation that from the decomposition
\begin{equation}\label{eq:hamDECOMPOSITION}
H_G = \sum_{i=1}^N h^{(i)}_G\ ,
\end{equation}
with $h^{(i)}_G =(x_i -x_0) \sigma^{(i)}_z /2$, we can simplify the QFI for product states $\state=\bigotimes_{i=1}^N \state_i$ via
\begin{equation}\label{eq:additivity}
\QFI\left[\bigotimes_{i=1}^N \state_i,H_G \right]=\sum_{i=1}^N \QFI\left[\state_i,h^{(i)}_G \right] \ ,
\end{equation}
by using the additivity of the QFI \cite{Toth2014}.
 Using  this relation together with the convexity of QFI and Eq.~\eqref{eq:simpl QFI} we find that the maximum of the QFI on the set of separable states on $N$ qubits $\mathrm{SEP}_N$ is obtained for the product state $\otimes_{i=1}^N \state_i$ such that each $\state_i$ maximizes $\QFI\left[\state_i,h^{(i)}_G \right]$. Therefore, we have 
\begin{equation}\label{eq:supremum separable}
\max_{\state\in\mathrm{SEP}_N}\QFI(\state)=(\gamma t)^2\sum_{i=1}^{N}\left(x_i -x_0\right)^{2}\,.
\end{equation}
and the maximum is obtained for the state 
 \begin{equation}
 \ket{\mathrm{P}}:=\ket{+}^{\otimes N},\,\,\,\mathrm{with}\,\,\,\ket{ +}=\frac{1}{\sqrt{2}}\left(\ket 0+\ket 1\right).\label{eq:def:product}
 \end{equation}
 From Eq.~\eqref{eq:supremum separable} we get the bound in precision for separable input states
 \begin{equation}\label{eq:LimitGSep}
 \Delta^2 \tilde{G}\geq \frac{1}{(\gamma t)^2\sum_{i=1}^{N}\left(x_i -x_0\right)^{2}} \ .
 \end{equation}
\newline
\textbf{Entangled states:} The second result concerns the maximal QFI over all states from the $N$ qubit Hilbert space $\H_N$. To compute this maximum we use  Eq.~\eqref{eq:simpl QFI}, Eq.~\eqref{eq:maxfisch}, and the fact that $H_G$ can be explicitly  diagonalized by the computational basis of the $N$ qubit Hilbert space $\H_N$. We obtain 
\begin{equation}\label{eq:max FisherG}
\max_{\state\in\D(\H_N)}\QFI(\state)=(\gamma t)^2\left[\sum_{i=1}^{N}\left(x_i -x_0\right)\right]^2\ ,
\end{equation}
with the optimal state being the $N$ qubit Greenberger-Horne-Zeilinger (GHZ) state \cite{Greenberger1989}
\begin{equation}
\ket{\GHZ}\coloneqq\frac{1}{\sqrt{2}} \left(\ket{0}^{\otimes N}+\ket{1}^{\otimes N}\right).\label{eq:def:GHZ}
\end{equation}
Analogously to the case of separable states in Eq.~\eqref{eq:LimitGSep} we use the quantum Cram\'{e}r-Rao bound to get the limitations on the precision for estimating $G$ with entangled states
\begin{equation}\label{eq:limitALL}
 \Delta^2 \tilde{G}\geq \frac{1}{(\gamma t)^2\left[\sum_{i=1}^{N}\left(x_i -x_0\right)\right]^{2}} \ .
\end{equation}

Let us remark that both, the maximal QFI in Eq.~\eqref{eq:max FisherG} and the maximal QFI for separable states in Eq.~\eqref{eq:supremum separable} strongly depend on the positioning of the particles and the coordinate $x_0$, if the value of the magnetic offset field $B_0$ is assumed to be known. Notice, however, that the quantum states for which the optimal values are attained do not depend on the spacing of particles. Moreover, the optimal states derived by us are invariant under the relabeling of qubits according to $B(x_i) \ge B(x_{i+1})$, which proves that our scheme works also for a negative value of the gradient $G$. Let us finally remark that in our analysis we have assumed, according to the note in the end of Section~\ref{sec:gradient_estimation_theory}, that $x_i \geq x_0$. The structure of optimal states and the precise formula for maximal QFI changes if this assumption is dropped. We discuss this in detail in Appendix~\ref{app:maxFI}. 
\newline
\textbf{Equidistant spacing:}
Neutral atoms in an optical microtrap are equidistant spaced. We consider an equidistant spacing in the interval $\left[x_0,L+x_0\right]$, i.e. $x_i-\!x_0\!=\!(\!i-\!1\!) \frac{L}{N-1}$ for measuring the gradient $G$ with $N$ qubits. For this positioning the QFI for separable states is given by
\begin{equation} \label{eq:linspec sep}
\max_{\state\in\mathrm{SEP}_N}\QFI(\state)= \frac{(\gamma t L)^2}{6}  \frac{N(2N-1)}{N-1}\, ,
\end{equation}
which (for fixed length $L$) scales proportionally to $N$ for a large number of particles.
On the other hand, the QFI for entangled states becomes
\begin{equation}
\max_{\state\in\D(\H_N)}\QFI(\state)=\frac{(\gamma t L)^{2}}{4} N^2\, ,\label{eq:linspec general}
\end{equation}
and scales with $N^2$  (for fixed $L$). 
\newline
\textbf{Optimal positioning and the ultimate bounds}:
We can optimize the QFI in Eq.~\eqref{eq:supremum separable} and Eq.~\eqref{eq:max FisherG} over the positioning $x_0$.
Again, we assume that the particles are located in the interval $\left[x_0,L+x_0\right]$ and we fix both $x_0$ and $L$. In order to maximize the right-hand sides of both Eq.~\eqref{eq:supremum separable} and Eq.~\eqref{eq:max FisherG},  an experimenter should put all qubits at the position $x_i=x_0+L$. This means that the particles are as far away as possible from the point $x_0$. If this is the case we get for separable states
\begin{equation}\label{eq:ultimate_max FisherG_sep}
\max_{\state\in\mathrm{SEP}_N}\QFI(\state)= (\gamma t)^2 N L^2.
\end{equation}
Similarly, the maximal QFI over all states becomes
\begin{equation}\label{eq:ultimate_max FisherG}
\max_{\state\in\D(\H_N)}\QFI(\state)=(\gamma t)^2 N^2 L^2. 
\end{equation}
An experimental realization of this positioning includes another dimension of the system. Atoms in an optical microtrap can be arranged in a $2$ dimensional lattice. Then, one dimension can be defined as the $x$ dimension and all qubits can be placed at one position $x=x_0+L$ by using the second dimension. In state of the art ion traps this arrangement is hard to realize. However, in future on-chip ion traps as proposed, e.g., in Refs.~\cite{Kielpinski2002,Lekitsch2017} this arrangement is possible.  In both types of experiments the extension of the qubit chain at $x_0+L$ must be much smaller as $L$ in order to exclude effects due to field gradients in the second dimension.

Using Eq.~\eqref{eq:LimitGSep} and Eq.~\eqref{eq:limitALL}  we can give now the ultimate bounds for the precision of estimating $G$, with the usage of $N$ qubits placed in the fixed interval $\left[x_0,L+x_0\right]$, and when we perfectly know the value of the field at $x_0$.
For separable probe states, the best achievable precision for the determination of $G$ is given by
\begin{equation}\label{eq:SQL_full}
\Delta^2 \tilde{G} \ge \frac{1}{(\gamma t)^2 N L^2},
\end{equation}
similar to the SQL.
Likewise, for entangled probe states, we get a Heisenberg-like scaling given by
\begin{equation}\label{eq:HL_full}
\Delta^2 \tilde{G} \ge \frac{1}{(\gamma t)^2 N^2 L^2}. 
\end{equation} 
For both, the scaling behavior in $N$ (for fixed {length $L$}) is identical to the case of the estimation of global parameters.  This is not surprising, since we assumed that we perfectly know the value of the offset field $B_0$ at  the position $x_0$ and the optimal strategy is to use all particles for the estimation of the field at the position $x_0+L$.

\subsection{Optimal measurements for experimental realizations}\label{sec:measurements_noiseless_case}

The optimal states derived by us can be prepared in  experimental settings such as trapped ions and neutral atoms in optical microtraps. In experiments with trapped ions the preparation of GHZ states \cite{Greenberger1989} up to $N=14$ qubits with high fidelity is possible \cite{Monz2011}. In experiments with neutral atoms in optical microtraps the preparation of a Bell or GHZ state as defined in Eq.~\eqref{eq:def:GHZ} with $N=2$ qubits has been achieved \cite{Wilk2010}.

However, as explained in Section~\ref{sec:gradient_estimation_theory} the bounds involving the QFI assume implicitly the application of the optimal measurement. In general, the optimal measurement saturating the quantum Cram\'{e}r-Rao bound is the projective measurement in the eigenbasis of the so-called symmetric logarithmic derivative \cite{Giovannetti2006}. This measurement can be difficult to perform in practice. Fortunately, the optimal measurement is not necessarily unique. In what follows we show that parity measurements in the $x$ basis are sufficient in order to reach the maximal possible precision in gradient estimation with GHZ states. Parity measurements can be easily performed in experiments with trapped ions \cite{Leibfried2004,Monz2011} as first proposed by Bollinger et al. in 1996 \cite{Bollinger1996} and neutral atoms in optical microtraps \cite{Isenhower2010,Wilk2010}.
A parity measurement is basically a detection of the number of qubits in either the spin-up or spin-down state and can be realized with almost $100\%$ efficiency \cite{Bollinger1996}. 
 Interestingly, the parity measurement does not depend on the spacing of particles and is thus the same for any configuration $\lbrace x_i \rbrace$.  
\newline
\textbf{Classical Fisher information:} 
A parity measurement in the $x$ basis is a projective measurement $ \PPP \coloneqq \{P_+,P_-\}$ with the projective operators
\begin{equation}\label{eq:POVM}
P_+ = \frac{1}{2}\left(\I +\sigma_{x}^{\ot N}\right)\ ,\  P_- =\frac{1}{2}\left(\I- \sigma_{x}^{\ot N}\right)\ .
\end{equation}
After the time $t$ the initial $N$-qubit state $\state$ evolves due to $U_G \rho U_{G}^\dagger$. Upon measuring $\PPP$ on $\rho_G$, the output probabilities are given by
\begin{equation}
p_+(G)= \tr\left(U_G \rho U_{G}^\dagger P_+\right)\ ,\  p_-(G)= \tr\left(U_G \rho U_{G}^\dagger P_-\right)\ .
\end{equation}
In Appendix~\ref{app:FI} we show that 
\begin{equation}\label{eq:EXPvalue}
\tr\left(U_G \psi_{\GHZ} U^\dagger_G  \sigma_{x}^{\ot N} \right)=\cos\left[ N \gamma B_0 t+ \gamma G t \sum_{i=1}^N (x_i-x_0) \right]\ ,
\end{equation}
where $\psi_{\GHZ}\coloneqq\kb{\GHZ}{\GHZ}$. Using this expression, together with Eq.  \eqref{eq:POVM} and the definition of the FI in Eq.~\eqref{eq:FI} we find that the classical Fisher information associated with the statistics of parity measurements with GHZ states is given by
\begin{equation}\label{eq:cfiGHZ}
\FI(U_G \psi_{\GHZ} U_{G}^\dagger,  \PPP) = (\gamma t)^2 \left[\sum_{i=1}^{N}(x_i-x_0)\right]^2 \ ,
\end{equation}
and equals the QFI for estimating $G$ with GHZ states [see Eq.~\eqref{eq:limitALL}]. Therefore, parity measurements in the $x$ basis are optimal for gradient estimation with GHZ states. 
The choice of an optimal measurement is not unique, also measurements of the collective spin operator in the $x$-direction $J_x$ are optimal as shown in Appendix~\ref{app:JXmes}. 
\newline
\textbf{Error propagation formula:}
It turns out that measurements of the expectation value of the parity $\hat{M}=P_++P_-=\sigma^{\ot N}_{x}$ with GHZ states also saturate the ultimate limitations given in Eq.~\eqref{eq:limitALL} for the accuracy of the measurement of  $G$. In usual experiments this expectation value $\braket{\hat{M}}$ is measured for different probing times $t$ and if the initial state is $\psi_{\GHZ}$ the theoretical time dependence is given in Eq.~\eqref{eq:EXPvalue}. In this measurement scheme the gradient $G$ is deduced from the value of the  frequency, which can be estimated by a fit on the data. This procedure however requires the known value of the offset field $B_0$. If one has no \emph{a priori} knowledge about $B_0$ one has to average \cite{Bartlett2007} over all possible values of $B_0$ and one cannot infer the value of $G$. It is possible to avoid this problem by measuring  this expectation value for different positioning $\{x_i\}$ at a \emph{fixed} probing time $t$. This strategy has been realized with a single ion moving through a gradient field in Ref.~\cite{Walther2011}. However, this scheme is definitely a less practical solution as one has to make sure that the initial quantum states are the same, despite the change in the configuration of the chain. The case of no \emph{a priori} knowledge about the offset field will be considered systematically in Section~\ref{sec:no_a_priori_knowledge}.

 With the error propagation formula in Eq.~\eqref{eq:error_propagation_formula} and using Eq.~\eqref{eq:EXPvalue} together with the fact that $\braket{\hat{M}^2}=1$ we can show that both measurement strategies (varying the probing time at a fixed positioning and varying the positioning at a fixed probing time) saturate the Cram{\'e}r-Rao bound and therefore also the analogue HL for gradient estimation, assuming that the positioning $\{x_i\}$ and the measurement time $t$ can be determined with high precision, so we have
\begin{equation}
\Delta^2 \tilde{G}_{\hat{M}}=\frac{1}{(\gamma t)^2\left[\sum_{i=1}^{N}\left(x_i -x_0\right)\right]^{2}} \ .
\end{equation}

\subsection{Gradient estimation in presence of collective phase noise}\label{sec:estimatin_with_noise}
\begin{figure}
\begin{center}
\includegraphics[width=0.45\textwidth]{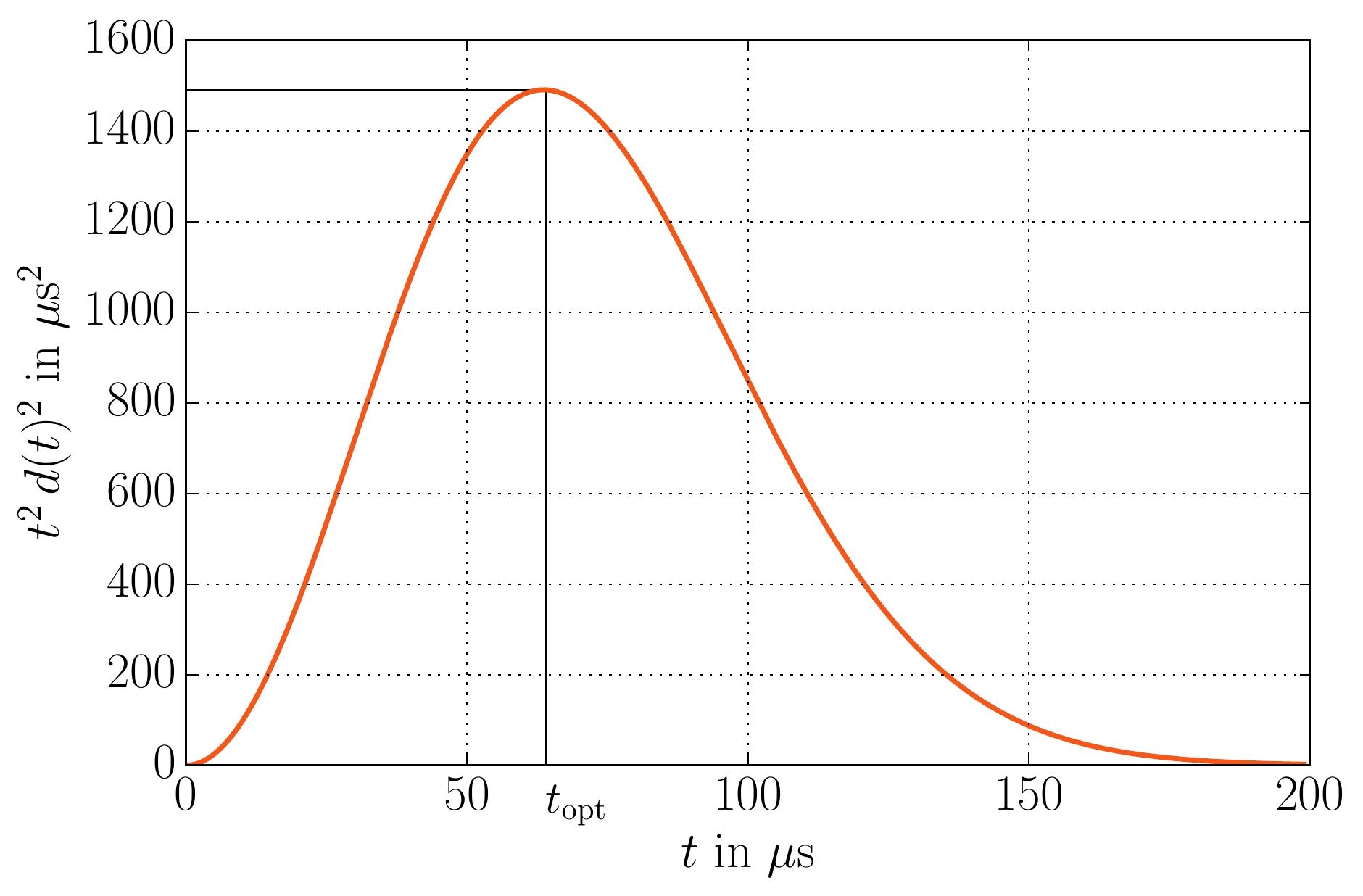}
\end{center}
\caption{The time dependent part of the QFI for a GHZ state with $N=50$ qubits in the presence of collective phase noise [see Eq.~\eqref{eq:QFI_noisy_GHZ}]. Here, we assume typical field fluctuations in the order of $\gamma' \Delta E=2 \pi\cdot 50\,$Hz and correlation time $\tau_c=1\,$s (see, e.g., Ref.~\cite{Baumgart2016}).
}\label{fig:noisyGHZ}
\end{figure}
In realistic experiments noise affects the time evolution of a quantum system,  reducing the entanglement of the probe states. This can diminish the enhancement in precision for  gradient estimation obtained with entangled states. In both types of experiments considered in this work collective phase noise is the main source of decoherence. In experiments with trapped ions, this noise is caused by temporal magnetic field fluctuations \cite{Monz2011}, whereas in experiments with neutral atoms in optical microtraps, collective phase noise is caused by temporal fluctuations of the trapping potential \cite{Kuhr2005}. In what follows we describe the influence of collective phase noise on the proposed scheme for the estimation of the gradient of the magnetic field.
\newline
\textbf{Collective phase noise:}
We focus our attention on trapped ions. We  follow the steps and assumptions for describing the noise source in this system given in Ref.~\cite{Monz2011}.
The total Hamiltonian of the system including the noise is given by
\begin{equation}
H'=\hbar \gamma H_0 + \hbar \gamma H_G +\hbar \gamma' \Delta E(t) J_z \ ,
\end{equation} 
where operators $H_0$ and $H_G$ are defined in Eq.~\eqref{eq:hamiltonian}, $\gamma'$ is the coupling constant, and $\Delta E(t)$ is the temporally fluctuating random field. We will use $\left\langle \cdot\right\rangle$ to denote the average over the stochastic fluctuations of this field. Following Ref.~\cite{Monz2011} we assume (i) no systematic time-dependent bias due to phase fluctuations  $\braket{\delta\varphi}=0$, where $\delta\varphi\coloneqq\int_0^t\mathrm{d}\tau \Delta E(\tau)$, (ii) Gaussian character of the fluctuations $\delta\varphi$,   (iii) stationarity of the noise process, $\braket{E(t+\tau)E(t)}=\braket{E(\tau)E(0)}$, and finally (iv) that the time correlation $\braket{\Delta E(t) \Delta E(0)}=(\Delta E)^2 \exp\left[-t/\tau_c \right]$ decays exponentially, with the correlation time $\tau_c$ and the fluctuation strength $\Delta E$.

Now, for a fixed realization of the stochastic process the output state at a given time $t$ is given by 
\begin{equation}
\state(t)=U'_{t} \state (U'_{t})^\dagger\ ,
\end{equation} 
with $U'_{t}=U_G U_\mathrm{noise}$, where $U_G$ is given in Eq.~\eqref{eq:timeevolution} and $U_\mathrm{noise}=\mathrm{exp}\left[-\ii\gamma' \int_0^t\Delta E(\tau) \mathrm{ d}\tau J_z\right]$ describes the noise acting on the system.
By averaging over the realization of the stochastic process $\Delta E(t)$ (for the fixed time $t$) we get that the initial state $\state$ is mapped into $\Lambda'_G(\state)$, where 
\begin{equation}\label{eq:NOISoutput}
\Lambda'_G(\state)= U_G \bar{\state}(t) U_G^\dagger\ \ \text{with}\ \ \bar{\state}(t)\coloneqq \langle U_{\mathrm{noise}} \state U^\dagger_{\mathrm{noise}}\rangle\ .
\end{equation}
 Since the encoding of the value of the gradient commutes with the map describing the noise we have
\begin{equation}\label{eq:NOISYqfi}
\QFI\left(\state,\Lambda'_G\right)=(\gamma t)^2\QFI\left[ \bar{\state}(t),H_G \right] \ .
\end{equation}
That is, in order to compute the QFI in the presence of collective phase noise it suffices to calculate the "standard" QFI on the noisy initial state.  In ion trap experiments, for which this paper is relevant, the repetition rate (i.e. the rate with which a single experiment can be repeated) typically is fixed  for noise cancellation of another noise source and $t\approx \mu$s-ms \cite{Monz2011}.  Moreover, the correlation time for the field fluctuations $\Delta E(t)$ is of order $\tau_c \approx$s . Therefore, we can assume $\tau_c \gg t$ \cite{Monz2011}.
\newline
\textbf{Noisy gradient estimation with GHZ states:}
For a GHZ state in presence of collective phase noise, the QFI can be calculated analytically (see Appendix~\ref{app:GHx_under_noise} for details) and takes a closed form
 \begin{equation}\label{eq:QFI_noisy_GHZ}
\QFI\left(\psi_{\GHZ},\Lambda'_G\right)=  d(t)^2 \gamma^2 t^2 \left[\sum_{i=1}^N (x_i-x_0)\right]^2\ ,
\end{equation}
with $d(t)=\mathrm{exp}\left\{-\left(N \gamma' \Delta E \tau_c\right)^2 \left[\exp(- \frac{t}{\tau_c})+\frac{t}{\tau_c} -1\right]\right\}$. 
We see that the QFI first increases with $t^2$ and then, in the limit of large times, decreases double exponentially to zero. Therefore, there exists a global maximum\footnote{Because the repetition rate is fixed, the relevant figure of merit to optimize is $\QFI(t)$ rather than $\QFI(t)/t$, which appears naturally when a variation of the repetition rate is possible and the total time of the experimental procedure is \emph{fixed}  \cite{Demkowicz-Dobrzanski2012,Escher2011}. } and an optimal measurement time as shown in Fig.~\ref{fig:noisyGHZ}.
 Under the condition $\tau_c \gg t$ [see the discussion below Eq.~\eqref{eq:NOISYqfi}] we get $\mathrm{exp}[(-t/\tau_c)+t/\tau_c -1]\approx 1/2(t/ \tau_c)^2$, which gives the optimal measurement time $t_\mathrm{opt}=\sqrt{2}/(N \gamma' \Delta E )$ and the maximal QFI 
\begin{equation}
\QFI=2 \gamma^2 \left[\sum_{i=1}^N (x_i-x_0)\right]^2 / e (N \gamma' \Delta E )^2 ,\label{eq:noise_opt_t}
\end{equation}
which can be maximized over the positioning $x_i$. The positioning maximizing the QFI in Eq.~\eqref{eq:noise_opt_t} leads to placing all qubits as far away as possible from $x_0$ that is $x_0+L$. Then, the maximal QFI is given by
\begin{equation}
\QFI=2 \gamma^2 L^2 / e ( \gamma' \Delta E )^2 ,
\end{equation}
which does neither scale with $N$ nor with $N^2$.

Now we will discuss the saturation of  the quantum  Cram\'{e}r-Rao bound for the estimation of the gradient $G$, with the QFI as in Eq.~\eqref{eq:QFI_noisy_GHZ} with parity measurements.
Generically we can achieve such a saturation by a suitable choice of the global phase $\theta$ for the probe state
\begin{equation}
\ket{\GHZ_\theta}\coloneqq \ \frac{1}{\sqrt{2}} \left(\ket{0}^{\otimes N}+\exp(\ii \theta)\ket{1}^{\otimes N}\right)\ ,\label{eq:GHZ_theta}
\end{equation}
 and performing a parity measurement
$\hat{M}=\sigma_x^{\ot N}$ as shown in Appendix~\ref{app:Parity_noisyGHZ} and \ref{app:EPF_noisyGHZ}.
Then, the Cram\'{e}r-Rao bound can be saturated if the condition
\begin{equation}
\mathrm{cot}\left[N\gamma B_0 t+\gamma G t \sum_{i=1}^N (x_i-x_0)+\theta\right]=0,
\end{equation}
holds. However, due to this condition an experimenter must have full knowledge about $B_0$ and $G$ at all measurement times $t$ in order to prepare the state in Eq.~\eqref{eq:GHZ_theta} that is not feasible.
\newline
\textbf{Noisy gradient estimation with the state $\ket{\mathrm{P}}$:}
In general, it is difficult to evaluate the QFI in Eq.~\eqref{eq:NOISYqfi} analytically  for arbitrary probe states $\state$. The initial pure product state $\psi_{\mathrm{P}}\coloneqq \kb{\mathrm{P}}{\mathrm{P}}$, given in Eq.~\eqref{eq:def:product}, evolves into a mixed state $\bar{\psi}_{\mathrm{P}}(t)$ due to noise. We will focus on the regime of large probing times $t\xrightarrow{} \infty$. In this limit, the state does not change any more due to collective phase noise and $[\bar{\psi}_{\mathrm{P}}(\infty),J_z]=0$ \cite{Lidar1998} and therefore we call this regime the \emph{steady state regime}. 
We get that for this state the QFI does not vanish in the limit of large times (see Appendix~\ref{app:product_state_steady} for details),
\begin{equation}
\QFI\left(\psi_{\mathrm{P}},\Lambda'_G\right)\xrightarrow{t\rightarrow\infty}  (\gamma t)^2\left[\sum_{i=1}^N x_i^2-\frac{1}{N} \left( \sum_{i=1}^N x_i\right)^2\right]\ .\label{eq:qfi_p_steady}
\end{equation}
Thus, in the steady state regime, the product state $\ket{\mathrm{P}}$ performs better then the GHZ state in the presence of noise. Interestingly the QFI in Eq.~\eqref{eq:qfi_p_steady} is independent of $x_0$.
This is due to the fact that in the steady state regime the probe state is not only invariant under collective phase noise but also under the offset field $B_0$ since $[\bar{\psi}_{\mathrm{P}}(\infty),J_z]=0$.
\newline
\textbf{Optimal positioning:}
For the GHZ state in the steady state regime the QFI vanishes independent of the positioning of the qubits.
However, for the separable state $\ket{\mathrm{P}}$ in the steady state regime the optimal positioning (see  Lemma 1 in Appendix~\ref{app:product_state_steady} for the proof) at the interval $x_i\in [\tilde{x}_0,\tilde{x}_0+L]$ is to place one half of the qubits at position $x_i=\tilde{x}_0$ and the other half as far away as possible $x_i=L+\tilde{x}_0$ (note that the position $\tilde{x}_0$ is some fixed reference coordinate and can have any value including $x_0$). For this case the QFI is given by
\begin{equation}\label{eq:qfi_p_steady_opt}
\QFI\left(\psi_{\mathrm{P}},\Lambda'_G\right)\xrightarrow{t\rightarrow\infty}  (\gamma t)^2 L^2 \frac{N}{4}\,,
\end{equation}
that is linear in $N$ and by constant factor of $1/4$ smaller than in the case of having no noise and placing all qubits at $x_i=x_0 + L$ [see Eq.~\eqref{eq:ultimate_max FisherG_sep}].  This can be realized by a similar arrangement as described in Section~\ref{sec:full_a_priori_knowledge}~A.

The optimal spacing of the particles leads to the situation in which the particles are located at two different positions. As a result, two different  unitaries $\UU_{G}^I$ and $\UU_{G}^{II}$ act on one half of the particles each. This is a local estimation strategy similar to differential interferometry \cite{Landini2014, Altenburg2016}. 
In these works phase and frequency estimation in the presence of correlated phase noise were investigated.  It was shown that the quadratic scaling in $N$  can be preserved by the usage of differential interferometry  in  presence of correlated noise. Furthermore, it was shown that for the product state $\ket{\mathrm{P}}$ a linear scaling in $N$ up to a constant factor can be preserved. In Eq.~\eqref{eq:qfi_p_steady_opt} we find a similar result. Furthermore, in Ref.~\cite{Dorner2013} also similar results where found. Here, the two unitaries $\UU_{G}^I$ and $\UU_{G}^{II}=\UU_{-G}^I$ act on half of the particles each. For this estimation scenario it was shown that the HL can be preserved in presence of correlated dephasing.


\section{Gradient estimation without \textit{a priori} knowledge about $B_0$} \label{sec:no_a_priori_knowledge}
\begin{figure*}
\begin{center}
\subfigure[ ]{\includegraphics[width=0.4 \textwidth]{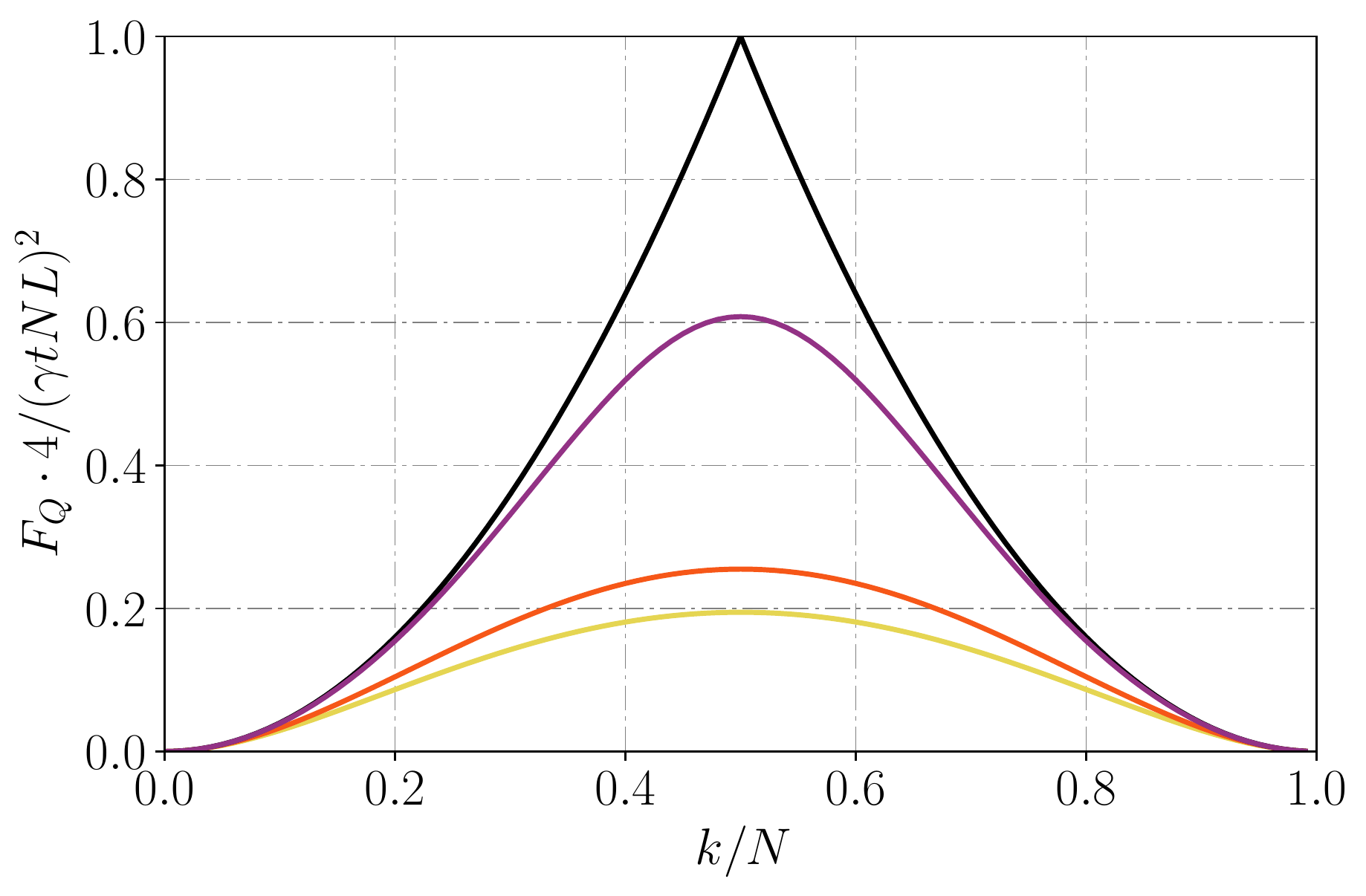}}
\subfigure[ ]{\includegraphics[width=0.5\textwidth]{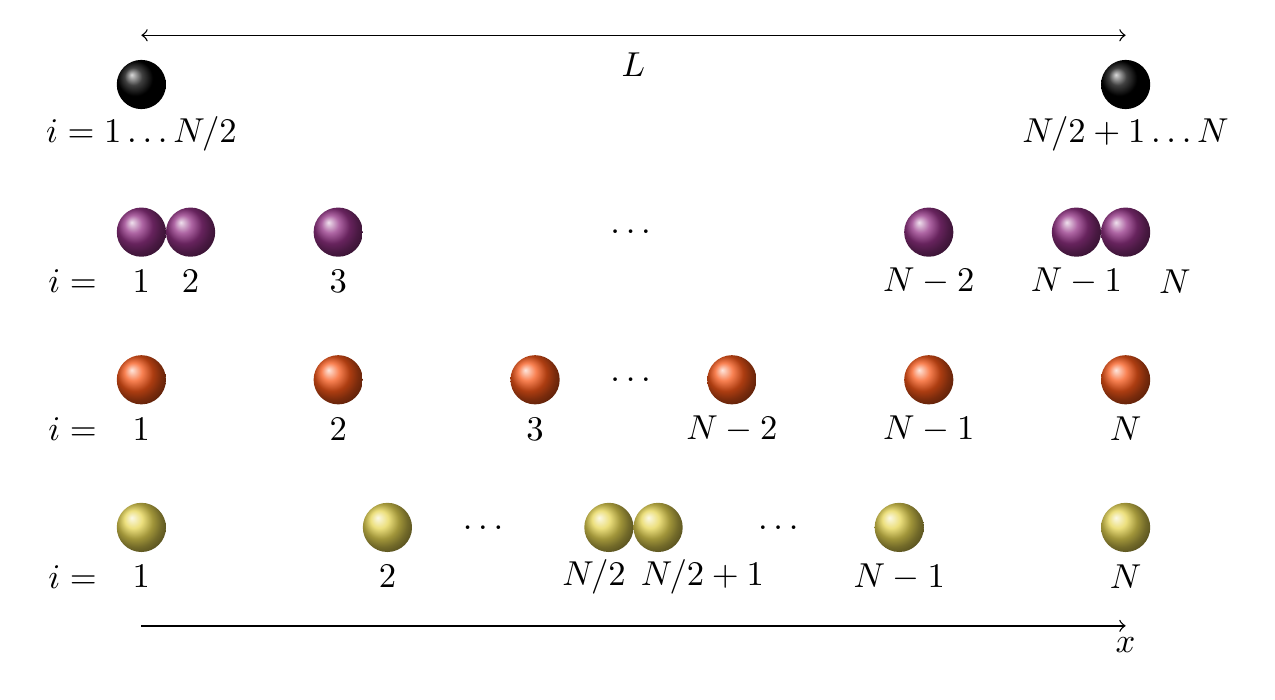}}
\end{center}
\caption{(Color online)
(a) Quantum Fisher information for optimal decoherence-free states $\ket{\ODF_k}$ for different  $k$, for $N=100$ qubits. Different grey shades (colors) denote different spatial distributions of particles.
(b) Different spatial distributions with a chain of $N$ qubits spread over the interval of length $L$.  Black denotes the optimal spatial distribution; $N/2$ qubits at one end of the chain and the other $N/2$ qubits at the other end of the chain. 
Darker grey (purple) denotes a distribution, where the qubits are more dense at the ends of the chain.
Grey (red) denotes an equidistant spatial distribution. 
Lighter grey (yellow) denotes a distribution, where the qubits are more dense in the middle of the chain. The spatial distributions used here are given in Appendix~\ref{app:spatial_distributions}. }\label{fig:max_qfi_vs_k}
\end{figure*}

In Section~\ref{sec:full_a_priori_knowledge}, we derived bounds in precision for the estimation of the gradient $G$, assuming complete knowledge of the offset field $B_0$. The question arises whether it is possible to measure $G$ without knowing anything about $B_0$. We can already answer this question:  collective phase noise can be interpreted as an erasure of information about $B_0$ in time. Therefore, waiting long enough leads to the case of having no knowledge about $B_0$. 
In Eq.~\eqref{eq:QFI_noisy_GHZ}  we saw that in the  steady state regime the QFI for the GHZ state $\ket{\GHZ}$  vanishes such that the GHZ state is useless in order to estimate $G$ in the case of having no knowledge about $B_0$.
However, in Eq.~\eqref{eq:qfi_p_steady_opt} we saw that in the  steady state regime the QFI for the product state $\ket{\mathrm{P}}$ didn't vanish. Therefore, it is possible to measure $G$ without knowing $B_0$. In this section we systematically study limits on the accuracy for estimating the gradient $G$, when no knowledge about $B_0$ is available.

In this section we first (A) identify optimal probe states and the corresponding bounds in precision for estimating $G$ when no \emph{a priori} knowledge of $B_0$ is available. Interestingly, we find that these bounds asymptotically behave in the similar way, as in the noiseless case. 
Then, (B) we prove  that similarly to the noiseless case parity measurements in the $x$ basis saturate the derived bounds in precision.  Finally, (C) we compare these results with the other measurement strategies considered in this paper and earlier works on the subject.

\subsection{Bounds in precision for gradient estimation}

In what follows we derive precision bounds for estimating a gradient with a fixed positioning $\{x_i\}$. Then, we discuss the case of equidistant spacing. Finally, we derive the optimal positioning for the particles located in the fixed interval of length $L$.

When assuming no \textit{a priori} knowledge about the offset field $B_0$ 
the Hamiltonian of the system does not change compared to Eq.~\eqref{eq:hamiltonian}. However, now the offset field $B_0$ is unknown and therefore must be treated as a random variable and all states, operations, and measurements performed on the system have to be averaged over all realizations of this random variable \footnote{Because of the specific form of the measurements and evolutions used in our analysis, we can limit ourselves only to averaging the initial  quantum states.}.  Phrasing this in a different language it can be said that we erase the reference frame \cite{Bartlett2007} associated to the  knowledge of the offset field [or, formally speaking, the  one-parameter group of transformations formed by operators $\exp\left(-\ii \theta J_z \right)$, where $\theta\in[0,2\pi)$]. Complete erasure of the knowledge about $B_0$ is modeled  by averaging the initial state $\state$ over all possible rotations around the $z$-axis
\begin{equation}\label{eq:averagedSTATE}
\bar{\state}\coloneqq \int_0^{2\pi} \frac{\mathrm{d}\theta}{2 \pi}\, \mathrm{e}^{-\ii 2 \theta  J_z} \state \mathrm{e}^{\ii 2\theta J_z}\ .
\end{equation}
States of the above form are called decoherence-free states since they are stationary states with respect to collective phase noise: $\left[\bar{\state},J_z \right]=0$. Conversely, every state $\tau$ satisfying   $\left[\tau,J_z \right]=0$ can be written as $\tau=\bar{\state}$, for a suitable $\state$ \cite{Lidar1998}. Decoherence-free states are insensitive to the offset field $B_0$ but in general can be affected by gradients i.e. $[\bar{\varrho}, H_G]\neq 0$. This suggests that they can be used for gradient estimation.
In what follows we will use  the decoherence-free subspace for $N$ qubits $\mathrm{DFS}_N$ to denote the set of decoherence-free states in the considered scenario. It is easy to see from the definition that every $\state\in\DFS_N$ can be written as a convex combination
\begin{equation}\label{eq:convDFS}
\state=\sum_{k=0}^N p_k \state_k
\end{equation}
of decoherence-free states $\state_k \in\D(\V_k)$, each  supported on the subspaces $\V_k$  spanned by computational basis vectors $\ket{i_1}\ket{i_2}\ldots \ket{i_n}$ containing exactly $k$ excitations ("1") and $N-k$ qubits in the ground state "0" \footnote{Alternatively, $\V_k$ can be characterized as the eigenspace of $J_z$ corresponding to the eigenvalue $\lambda_k =\frac{1}{2}(N- 2k)$.}. 
\newline
\textbf{Optimal decoherence-free states:} 
In order to compute how useful decoherence-free states are for the estimation of the gradient $G$ we use directly Eq.~\eqref{eq:simpl QFI}, which reduces the problem of computing the QFI for the proposed metrological scheme to the computation of the $\QFI[\state,H_G]$, where $H_G = \sum_{i=1}^N (x_i-x_0) \frac{\sigma_z^{(i)}}{2}$. Using the fact that $H_G$ preserves subspaces $\V_k$ and the properties of $\QFI[\state,H_G]$  (see Appendix~\ref{app:presicion_bounds_not_knowing_b} for details), we prove that in order to find optimal decoherence-free states it suffices to look only at optimal (and thus necessary pure) states in each subspace $\V_k$ separately, 
\begin{equation}
\max_{\state \in \DFS_N} \QFI(\state) = (\gamma t)^2 \max_{k=0,\ldots,N}\, \max_{\state\in\D(\mathcal{V}_k)} \QFI[\state,H_G] \,.\label{eq:sequentialreduction}
\end{equation}
The maximal attainable QFI for states in the subspaces $\V_k$ with $k$ excitations is given by (see Appendix~\ref{app:presicion_bounds_not_knowing_b})
\begin{equation}\label{eq:maxfischsubspace}
\max_{\state\in\D(\mathcal{V}_k)} \QFI(\state)=(\gamma t)^2\left[\sum_{i=1}^{l}\left(x_{i}-x_{N-i+1}\right)\right]^{2}\ ,
\end{equation}
where $l=\min\{k,N-k\}$. We observe that the above result is independent of $x_0$ and that it does not change under the simultaneous translation of each $x_i\xrightarrow{} x_i+\delta$ by the same distance $\delta$. This is a consequence of the relation  $\left[{\varrho},J_z\right]=0$, valid for $\state\in\DFS_N$. The optimal decoherence-free (ODF) state for a given number of excitations $k$, yielding Eq.~\eqref{eq:maxfischsubspace} is given by
\begin{equation}\label{eq:optstatesubspace}
\ket{\ODF_k}=\frac{1}{\sqrt{2}}\left(\ket{1}^{\otimes k}\otimes \ket{0}^{\otimes N-k}+\ket{0}^{\otimes N-k}\otimes \ket{1}^{\otimes k}\right)\ .
\end{equation}
The detailed derivation of  Eq.~\eqref{eq:maxfischsubspace}  and Eq.~\eqref{eq:optstatesubspace} is given in Appendix~\ref{app:presicion_bounds_not_knowing_b}.
A remarkable fact is that for decoherence-free states the QFI for the estimation of  the gradient $G$ does not decrease in time due to the collective phase noise. In Fig.~\ref{fig:max_qfi_vs_k} we show the QFI from Eq.~\eqref{eq:maxfischsubspace} for different number of excitations $k$ and for different positioning $\{x_i\}$. We observe that the maximal QFI is attained exactly for $k=N/2$, for an arbitrary positioning of the qubits\footnote{For simplicity we assumed that $N$ is even. In general the maximal QFI is attained for  $k=\lfloor N/2 \rfloor$ (see Appendix \ref{app:presicion_bounds_not_knowing_b} for details.)}. This observation can be proven analytically for any positioning of the particles (see  Appendix~\ref{app:presicion_bounds_not_knowing_b} for details) and we find
\begin{equation}\label{eq:optimalDFS}
\max_{\state \in \DFS_N} \QFI(\state)=(\gamma t)^2\left[\sum_{i=1}^{N/2 }\left(x_i-x_{N-i+1}\right)\right]^{2}\  
\end{equation}
with $\ket{\ODF_{N/2}}$ being the optimal state. It is important to note that just like in the noiseless case [see Eq.~\eqref{eq:def:GHZ}], the optimal state does not depend on the spacing of particles. From the quantum Cram\'{e}r-Rao bound in Eq.~\eqref{eq:Cramér_rao2} we get the ultimate bound on the precision of the estimation of the gradient $G$ with decoherence-free states
\begin{equation}\label{eq:precODF}
 \Delta^2 \tilde{G}\geq \frac{1}{(\gamma t)^2 \left[\sum_{i=1}^{N/2}\left(x_i-x_{N-i+1}\right)\right]^{2}} \ .
 \end{equation}

\noindent
\textbf{Separable states:} 
In the case of no \textit{a priori} knowledge about the offset field it is hard to derive precision bounds for separable states. This follows from the difficulty to characterize the convex set $\DFS_N\cap\mathrm{SEP}_N$, consisting of states that are both decoherence-free and separable. In particular, extremal points of $\DFS_N\cap\mathrm{SEP}_N$  generally do not have the form of pure state. We leave the problem of finding the optimal decoherence-free separable state open. However, let us remark that the decoherence-free separable state\footnote{Recall that the averaging operation $\bar{\state}(t)$ preserves separability of quantum states.} $ \bar{\psi}_P(\infty)$ exhibits asymptotically  the same (linear in $N$) scaling   of the QFI   as the optimal product state $\psi_{\mathrm{P}}= \kb{\mathrm{P}}{\mathrm{P}}$ at least for the case of equal and optimal spacing (that is placing half of the qubits at each position $\tilde{x}_0$ and $\tilde{x}_0+L$ for $ \bar{\psi}_P(\infty)$ and placing all qubits at position $x_0+L$ for $\psi_P$) - see Eq.~\eqref{eq:qfi_p_steady} and Eq.~\eqref{eq:qfi_p_steady_opt}.  
\newline
\textbf{Equidistant spacing:} 
Just like in the case of complete knowledge about $B_0$ (described in Sec.~\ref{sec:full_a_priori_knowledge}) we consider a measurement scheme in which $N$ particles are equally spaced in the interval  $\left[\tilde{x}_0,\tilde{x}_0+L\right]$, i.e.  $x_i=\tilde{x}_0+(i-1)\frac{L}{N-1}$ (recall that the position $\tilde{x}_0$ is some fixed reference coordinate and can have any value including $x_0$).  Then, for the optimal decoherence-free state $\psi^{N/2}_{\ODF}\coloneqq\kb{\ODF_{N/2}}{\ODF_{N/2}}$ we have
\begin{equation}
\QFI(\psi^{N/2}_{\ODF})=
\frac{(\gamma t L)^2}{16}\frac{N^4}{\left(N-1\right)^{2}}\ ,\label{eq:linspace decoh}
\end{equation}
which scales $\propto N^2$ for large numbers of particles.

With the optimal separable state from the noiseless case $\ket{ \mathrm P}$,  
the QFI for equidistant spacing in the steady state regime becomes
\begin{equation}\label{eq:product_equi}
\QFI\left(\bar{\psi}_P(\infty)\right)=\frac{(\gamma t L)^2}{12}\frac{N(N+1)}{N-1} \ , 
\end{equation}
which scales $\propto N$ for large numbers of particles.
\newline
\textbf{Optimal positioning:}
Optimizing the right hand side of Eq.~\eqref{eq:optimalDFS} over the positions $x_i \in [ \tilde{x}_0,\tilde{x}_0+L]$, we find the maximal QFI over all decoherence-free states
\begin{equation}\label{eq:max decoh fiisher location}
\max_{\state \in \DFS_N} \QFI(\state)=\frac{(\gamma t L)^2}{4}N^2\ ,
\end{equation}
which is independent of $\tilde{x}_0$ and scales \ $\propto N^2$. The optimal positioning leading to Eq.~\eqref{eq:max decoh fiisher location} is $x_i= \tilde{x}_0$ for $i\le N/2$ and $x_i=\tilde{x}_0 + L$ for $i> N/2$ (see Appendix~\ref{app:presicion_bounds_not_knowing_b} for the proof). This corresponds to locating the particles at two positions with the maximal possible distance $L$. Recall that the same positioning was found to be optimal for estimating the gradient $G$ with the state $\ket{\mathrm{P}}$ in the steady state regime [$\bar{\psi}_P(\infty)$ as discussed above Eq.~\eqref{eq:qfi_p_steady_opt}].

\subsection{Optimal measurements for the experimental realization}
Optimal decoherence-free states $\ket{\ODF_k}$ are equivalent under local unitaries to GHZ states, that means that ODF states can be transformed into GHZ states and vice versa by local unitaries.
ODF states can be prepared with high fidelity by a global S\o{}rensen-M\o{}lmer gate \cite{Sorensen1999} in experiments with trapped ions. That has been performed for $N=14$ qubits in Ref.~\cite{Monz2011}. In experiments with neutral atoms in a lattice the preparation of the ODF state with $k= N/2$ for $N=2$ qubits has been realized \cite{Isenhower2010,Wilk2010}.
We therefore conclude that optimal probe states for gradient estimation can be realized in experiments considered in this work. As in the case of full knowledge about $B_0$ and the absence of noise, the question remains of which measurement should be performed in order to attain the maximal precision. In what follows we show that  for parity measurements in the $x$ basis (i) the classical Fisher information [see Eq.~\eqref{eq:FI}] and   (ii) the error propagation formula [see Eq.~\eqref{eq:error_propagation_formula}] saturate the quantum Cram\'{e}r-Rao bound in Eq.~\eqref{eq:precODF} for optimal decoherence-free states $\ket{\ODF_k}$.
\newline
\textbf{Classical Fisher information:} 
As described in Sec.~\ref{sec:measurements_noiseless_case},
the projective measurement  $\PPP$ of the parity in the $x$-basis is described by the projectors $P_\pm =\frac{1}{2}\left(\I \pm\sigma_{x}^{\ot N}\right)$ [see Eq.~\eqref{eq:POVM}]. In Appendix~\ref{app:FI} we show that the expectation value of the parity on the state $\psi^k_{\ODF}$ evolves according to
\begin{equation}\label{eq:expODF}
\tr\left(U_G \psi^k_{\ODF} U^\dagger_G  \sigma_{x}^{\ot N} \right)=\cos\left[\gamma t G \sum_{i=1}^{l }\left(x_i-x_{N-i+1}\right) \right] \ ,
\end{equation}
where $l=\min\lbrace k, N-k\rbrace$. Using this result and performing the analogous computations as the ones given in Sec.~\ref{sec:measurements_noiseless_case}  we get
\begin{equation}\label{eq:CFIodf}
\FI(U_G \psi^k_{\ODF} U_{G}^\dagger,  \PPP)=  (\gamma t)^2 \left[\sum_{i=1}^{l}\left(x_i-x_{N-i+1}\right)\right]^{2} \ ,
\end{equation}
with $l=\min\lbrace k, N-k\rbrace$. Comparing Eq.~\eqref{eq:CFIodf} and Eq.~\eqref{eq:optimalDFS} we see that parity measurements in the $x$ basis saturate the quantum Cram\'{e}r-Rao bound for the estimation of $G$ with the optimal decoherence-free state $\ket{\ODF_k}$.  In particular, the quantum Cram\'{e}r-Rao bound is saturated for the optimal state $\ket{\ODF_{N/2}}$ and therefore also the bound in Eq.~\eqref{eq:precODF} is saturated.
\newline
\textbf{Error propagation formula:} Just like in the noiseless case (discussed in Sec.~\ref{sec:measurements_noiseless_case}), one can try to estimate the gradient $G$ from the measurements of the expectation value  of $\hat{M}=\sigma_{x}^{\ot N}$ given in Eq.~\eqref{eq:expODF}. Using the error propagation formula in Eq.~\eqref{eq:error_propagation_formula} and the formula for the expectation value in Eq.~\eqref{eq:expODF} with $\braket{\hat{M}^2}=1$, we obtain that this measurement strategy again leads (for small fluctuations of the gradient $G$) to the maximal achievable precision, provided by the optimal state $\ket{\ODF_{N/2}}$:
\begin{equation}\label{eq:epfDFS}
\Delta^2 \tilde{G}_{\hat{M}}=\frac{1}{(\gamma t)^2 \left[\sum_{i=1}^{N/2}\left(x_i-x_{N-i+1}\right)\right]^{2}} \ .
\end{equation}

From the above discussion we see that parity measurements in the $x$ basis  are optimal  in both extremal scenarios considered in this paper - under the condition of full and no \emph{{a} priori} knowledge about the offset field $B_0$. As discussed in Sec.~\ref{sec:measurements_noiseless_case} these measurements can be routinely realized in experiments with trapped ions and neutral atoms in an optical lattices. It is also important that the ultimate accuracy given in Eq.~\eqref{eq:epfDFS} is saturated independently on the value of the gradient $G$ and does not deteriorate with time due to collective phase noise.

\begin{figure*}[]
\begin{center}
	\subfigure[ ]{\includegraphics[width=0.45\textwidth]{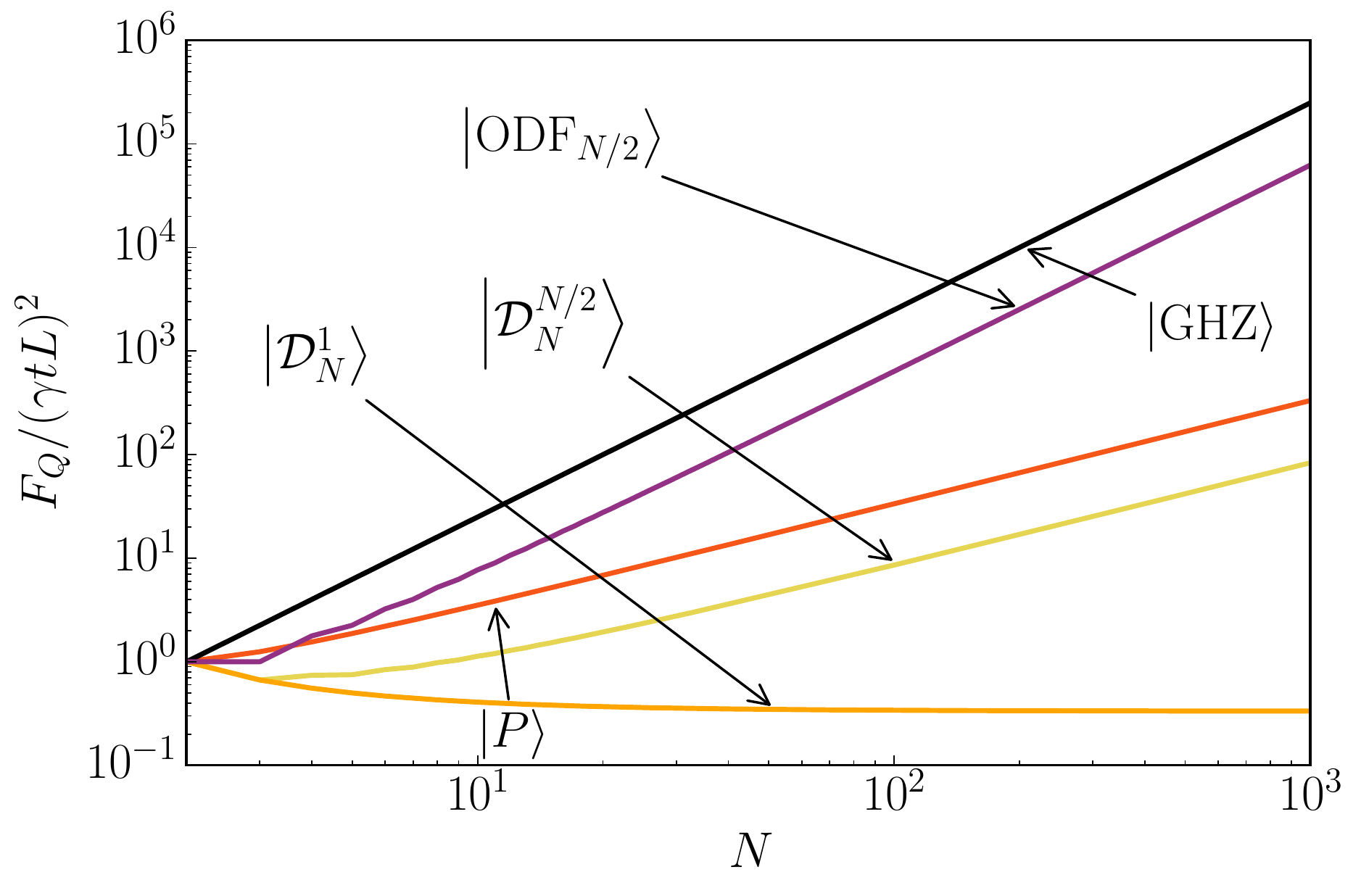}}
	\subfigure[ ]{\includegraphics[width=0.45\textwidth]{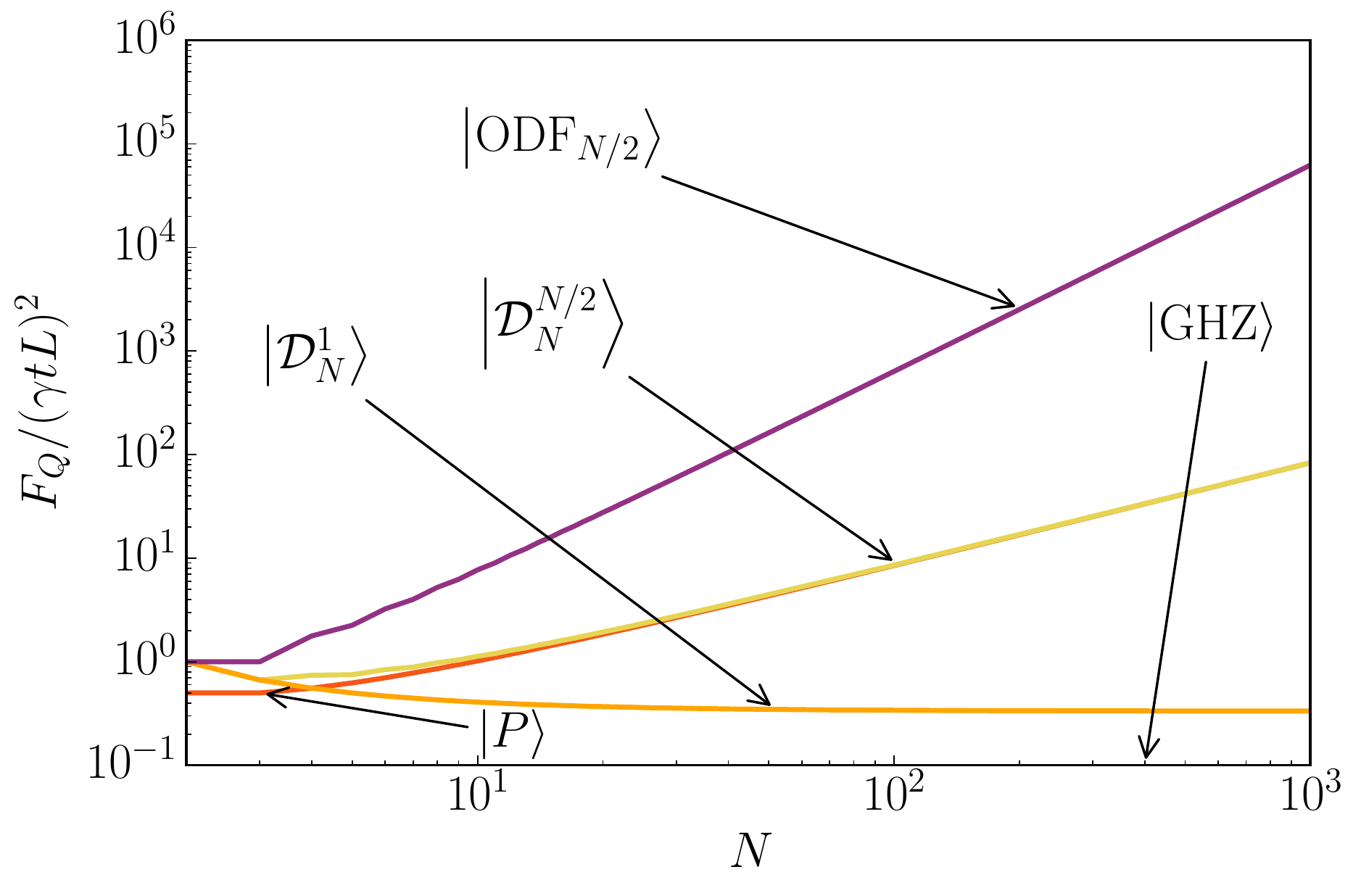}}
\end{center}
	\caption{(Color online) Quantum Fisher information for different families of quantum states for the estimation of the gradient of a magnetic field, under the assumption of equidistant spacing of the particles.
	The case of full \textit{a priori} knowledge of $B_0$ is depicted in part (a), whereas  the case of no \textit{a priori} knowledge of $B_0$ is given in part (b). 
		In (a) the GHZ state performs best and in (b) the QFI for GHZ states vanishes $\QFI=0$. The QFI for decoherence-free states is the same in both scenarios. 	The performance for optimal Dicke states $\ket{\D_N^{N/2}}$ is given by the solid yellow line, for W states $\ket{\D_{N}^1}$ by the solid orange line, and finally, for the optimal decoherence-free states  with $k=N/2$ by the solid purple line. The last one is not smooth for small $N$. This is due to the fact that the exact formula for the QFI differs from Eq.~\eqref{eq:linspace decoh} for odd $N$ as shown in Appendix~\ref{app:presicion_bounds_not_knowing_b}.
		In the presence of noise, both the product state $\ket{\mathrm{P}}$ (solid red line) and the GHZ state (solid black line) are between the QFI shown in (a) and the QFI shown in (b).  The slope of the curves represent the scaling in $N$, i.e. a $N^2$ scaling (analogue of HL), a linear scaling in $N$ (analogue of SQL) and a constant scaling in $N$ (W states $\ket{\D_{N}^1}$ for large $N$).	}\label{fig:scaling_equidistant}
\end{figure*}

\subsection{Comparison of the performance of decoherence-free states with other strategies}\label{sec:compare_strategies}

In this part we compare the performance of ODF states with (i) GHZ states in the case of full \textit{a priori} knowledge of $B_0$ but in the presence of collective phase noise, and (ii) with  Dicke states, for the estimation of gradients.
\newline
\textbf{Comparison with GHZ states:}
We compare the performance of GHZ states with optimal decoherence-free states when we have  complete information about $B_0$ and collective phase noise is present. As mentioned before, collective phase noise can be interpreted as an erasure of information about $B_0$. Therefore, the QFI for $\GHZ$ states vanishes for long measurement times $t$ and ODF states  $\ket{\ODF_{N/2}}$ with $k=N/2$ excitations perform better. In contrast, for short probing times $t$ GHZ states perform better. We can calculate the critical time $t_\mathrm{crit}$ for which the QFI for GHZ states under noise is equal to the QFI for $\ket{\ODF_{N/2}}$. In experiments for which this paper is relevant typically $t\approx\mu$s-ms and the correlation time $\tau_c\approx$s of the field fluctuations $\Delta E(t)$. Therefore, we can assume $\tau_c\gg t$ from which we get  (see Appendix~\ref{app:critical_time})
\begin{equation}
t_\mathrm{crit}=\frac{\left\{2 \log\left[ \frac{\left(\sum_{i=1}^N (x_i-x_0)\right)^2}{\left(\sum_{i=1}^{N/2}\left(x_i-x_{N-i+1}\right)\right) ^{2}}\right]\right\}^{1/2}}{N\gamma' \Delta E}\,,
\end{equation}
In the case of complete \textit{a priori} knowledge about $B_0$ at the beginning and collective phase noise GHZ states perform good for $t< t_\mathrm{crit}$. For $t>t_\mathrm{crit}$ ODF states $\ket{\ODF_{N/2}}$ outperform GHZ states.

Neutral atoms in an optical microtrap are arranged equidistant $x_i-x_0=(i-1) L/(N-1)$. For this positioning we find the critical time $t_\mathrm{crit}=2\sqrt{\log[2(N-1)/N]}/(N \gamma' \Delta E)$. This is independent of the total length $L$ of the string.
For $N=50$ qubits and $\gamma' \Delta E= 2 \pi\cdot 50\,$Hz (as used for Fig.~\ref{fig:noisyGHZ}) we find $t_\mathrm{crit}=104\,\mu$s and for $N=8$ we find $t_\mathrm{crit}=595\,\mu$s. Both are within typical coherence times of such experiments.

In Sec.~\ref{sec:full_a_priori_knowledge} we discussed the case of full \textit{a priori} knowledge about $B_0$.
Here, we found that in the absence of noise GHZ states are optimal. However, this holds only under the assumption that  $x_i\ge x_0$, this means that the whole string of qubits is on the right hand side of the position $x_0$ where $B_0$ is known. If $x_0$ is defined to be located within the range of the qubit string GHZ states are not optimal anymore (as shown in Appendix~\ref{app:maxFI}).
The experimentally relevant case is, if an experimenter has full \textit{a priori} knowledge about the offset field $B_0$
 right in the middle of the qubit string $x_{N/2}\le x_0 <x_{N/2+1}$, e.g., by estimating the average field. Interestingly, in this case the optimal states are ODF states (as shown in Appendix~\ref{app:maxFI}) that are decoherence free. Therefore, when $x_0$ is defined to be in the middle of the string it doesn't matter whether an experimenter has knowledge about the offset field. 
 This fact implies, that only if an experimenter is able to measure the offset field at a position that is not right in the middle of the string, she could gain from having information about the offset field.

 In principle \textit{a priori} knowledge about the offset field could enhance the precision for gradient estimation since the maximal QFI when having full  \textit{a priori} knowledge about the offset field in Eq.~\eqref{eq:ultimate_max FisherG}
 is by constant factor of $4$ greater than the one in the case of having no \textit{a priori} knowledge about the offset field in Eq.~\eqref{eq:max decoh fiisher location}.
However, this comparison is unfair because this enhanced precision is gained by an unknown amount of resources that was previously used to determine the off-set field $B_0$.  Furthermore, as discussed before a gradient measurement does not always gain from having \textit{a priori} knowledge about the offset field. In fact only for $t<t_\mathrm{crit}$ knowledge about the offset field enhances the precision for gradient estimation since collective phase noise immediately erases the information about the offset field.
 However, even for $t<t_\mathrm{crit}$ the gain from measuring the offset field is only a constant factor (up to $4$).
This factor can also be reached by using longer measurement times since $\QFI \propto t^2$ for ODF states (that are insensitive to the offset field). In the here considered experiments $t_\mathrm{crit}\propto\,\mu$s whereas typical measurement times $t \propto\,$ms such that $t>t_\mathrm{crit}$ as we discussed above. Then, ODF states perform better than GHZ states. Therefore, in the here considered experiments  it is not worth to spend any resources for a measurement of the offset field  for the estimation of gradients.

One possible objection to the above reasoning is that for any specified probing time $t$ there exist in principle optimal states and measurements that would give a precision for gradient estimation higher than the one for optimal DFS states [given in Eq.~\eqref{eq:precODF}]. The technical limitation of such a scheme is that the optimal states and measurements depend on the probing time which results in experimental difficulties. On the other hand, the optimal DFS states    and the corresponding measurements have already been implemented in  experiments \cite{Monz2011}.
\newline

\textbf{Gradient estimation with Dicke states:}
In Ref.~\cite{Zhang2014} it was claimed that for gradient estimation with a W state a good scaling of the QFI in $N$ is possible. Recall that W states are decoherence-free states  and belong to the set of symmetric Dicke states \cite{Dicke1954} $\ket{\D^{k}_N}= \frac{1}{\NN}\sum_j \PP_j\{\ket{0}^{\otimes N-k}\otimes\ket{1}^{\otimes k}\}$,
where $\NN$ is a normalization constant and  $\sum_j \PP_j\{.\}$ denotes the sum over all possible permutations. Symmetric Dicke states with $k=1$ excitations are exactly W states.
We use Eq.~\eqref{eq:simpl QFI} to compute the QFI for Dicke  states $\D_{N}^k \coloneqq \kb{\D^{k}_N}{\D^{k}_N}$ [see Appendix~\ref{app:product_state_steady}, Eq.~\eqref{app:eq:QFI_Dicke} for details] and the final result is    
\begin{align}
\begin{split}
\QFI(\D_{N}^k)/(\gamma t)^2
&= \sum_{i=1}^N x_i^2-\left( \sum_{i=1}^N x_i\right)^2 \frac{\left(2k-N\right)^2}{ N^2}\\&+\sum_{i \neq j=1}^N x_ix_j\left[\frac{(2k-N)^2-N}{N(N-1)}\right]\ .
\end{split}
\end{align}
For equidistant positioning $x_i=(i-1)\frac{L}{N-1}$ we have
\begin{align}
\QFI(\D_{N}^k)= (\gamma t)^2 \left(\frac{L}{N-1}\right)^2  \frac{(N+1) k (N-k)}{3},
 \end{align}
which is maximal for $k=N/2$,   
  \begin{align}
\QFI(\D_{N}^{N/2})=(\gamma t)^2\left( \frac{ L}{N-1}\right)^2  \frac{N^2 (N+1)}{12}\,.
  \end{align}
For W states ($k=1$) we have
\begin{align}\label{eq:QFI_W}
\QFI(\D_{N}^{1})= (\gamma t)^2 \left(\frac{L}{N-1}\right)^2  \frac{N^2-1}{3}\ .
 \end{align}   
This is exactly the same result as the one from Ref.~\cite{Zhang2014} with $a=L/(N-1)$. In Ref.~\cite{Zhang2014} $a$ is defined as a fixed distance between the qubits with $x_i=(i-1) a$, such that adding a qubit leads to an extension of the total length $L$ of the string. Using this convention the QFI for W states in Eq.~\eqref{eq:QFI_W} scales with $\QFI \propto a^2 N^2$ for large $N$. At first sight this seems to be a good scaling since it is quadratic in $N$. However, when fixing the distance between the qubits $a$ the HL from Eq.~\eqref{eq:HL_full} for gradient estimation is $\Delta^2 \tilde{G} \propto 1/N^4$ and the SQL from Eq.~\eqref{eq:SQL_full} is $\Delta^2 \tilde{G}  \propto 1/N^3$ for large $N$ and with $L=(N-1)a$. Therefore, a quadratic scaling in $N$ is not a good scaling for a fixed distance between the qubits.
Furthermore, when fixing the total length $L$ the QFI for W states decreases with $N$ to a constant $\QFI\left(\D_{N}^{1}\right)\xrightarrow{N\rightarrow\infty}  (\gamma t)^2 L^2 /3$.
The product state $\ket{\mathrm P}$ in the steady state regime in Eq.~\eqref{eq:product_equi}
 performs better then a W state in Eq.~\eqref{eq:QFI_W} and is for large $N$ equal to the maximal attainable QFI with symmetric Dicke states ($k=N/2$).

 We conclude the section with the graphical comparison of the performance of different families of states for  gradient estimation in Fig.~\ref{fig:scaling_equidistant}, under the assumption of (a) full or (b) no \emph{a priori} knowledge about the offset field $B_0$.

\section{Generalization} \label{sec:generalization} 
The model described in Section~\ref{sec:gradient_estimation_theory} can be generalized to an arbitrary known spatial distribution $f(x)$ of the $z$ component of the magnetic field. We can consider an experiment for the estimation of the strength $G$ of a spatial magnetic field distribution given by 
\begin{equation}
B(x)=B_0+G f(x-x_0) \ ,
\end{equation}
where the function $f(x)$ is known and $f(x_0)=0$ holds. E.g., due to the quadratic Zeeman effect it may be known that the field has to be quadratic in $x$ such that $f(x-x_0)=(x-x_0)^2-(x-x_0)a$ as depicted in Fig.~\ref{fig:particles_in_field}.
The Hamiltonian in Eq.~\eqref{eq:hamiltonian} is then generalized by replacing $x_i-x_0$ by $f(x_i-x_0)$, for $i=1,\ldots,N$. The labeling of the particles is then imposed by the ordering of the values of the magnetic field i.e.  $B(x_i)\le B(x_{i+1})$. Under these slight modifications essentially all the results presented in this paper carry over. In particular, the bounds on the precision given by the QFI in Eq.~\eqref{eq:max FisherG} and Eq.~\eqref{eq:supremum separable} for the case of full \textit{a priori} knowledge of $B_0$ and in  Eq.~\eqref{eq:optimalDFS} for the case of no \textit{a priori} knowledge are valid in this generalized model.

\begin{figure}[]
\begin{center}
\includegraphics[width=0.45\textwidth]{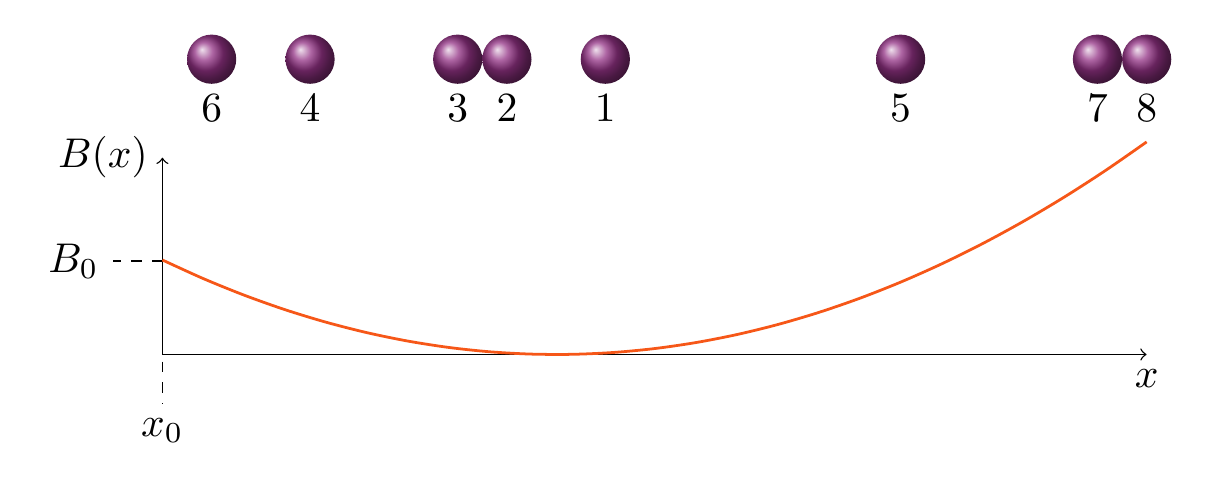}
\end{center}
\caption{A string of particles in a magnetic field with a spatial distribution along the string.  The magnetic field $B(x)$ acts on each particle, depending on its position. The particles are labeled such that the smallest magnetic field acts on the first and the highest magnetic field on the last particle.
 }\label{fig:particles_in_field}
\end{figure}

Also the optimal states and the optimal measurements attaining these bounds do not change.
Furthermore, the optimal positioning for the case of full \textit{a priori} knowledge about $B_0$ is to put all qubits at the position $x_{\max}$ that maximizes $f(x-x_0)$. For the case of no \textit{a priori} knowledge about $B_0$ it is optimal to put half of the qubits at the position $x_{\min}$ that minimizes $f(x-x_0)$ and the other half of the qubits at the position $x_{\max}$ that maximizes $f(x-x_0)$.
\newline
Note that one has to keep in mind that the above analysis is valid under the assumption $f(x-x_0)\geq 0$ (which corresponds to the condition $x_i \geq x_0$ from the note given in the end of Sec.~\ref{sec:gradient_estimation_theory}). 
As depicted in Fig.~\ref{fig:particles_in_field}  the assumption $f(x-x_0)\geq 0$ does in general not always hold.
If this assumption is dropped all the results for the case of no \textit{a priori} knowledge about $B_0$ (decoherence free states) still carry over. However, for the case of full \textit{a priori} knowledge about $B_0$ the precise form of the optimal states and formulas for maximal QFI change.  Although one can still recover   the results by substituting $x_i-x_0$ by $f(x_i-x_0)$ in the appropriate formulas given in Appendix~\ref{app:maxFI}.

\section{Conclusions}
\begin{figure*}
\includegraphics[width=\textwidth]{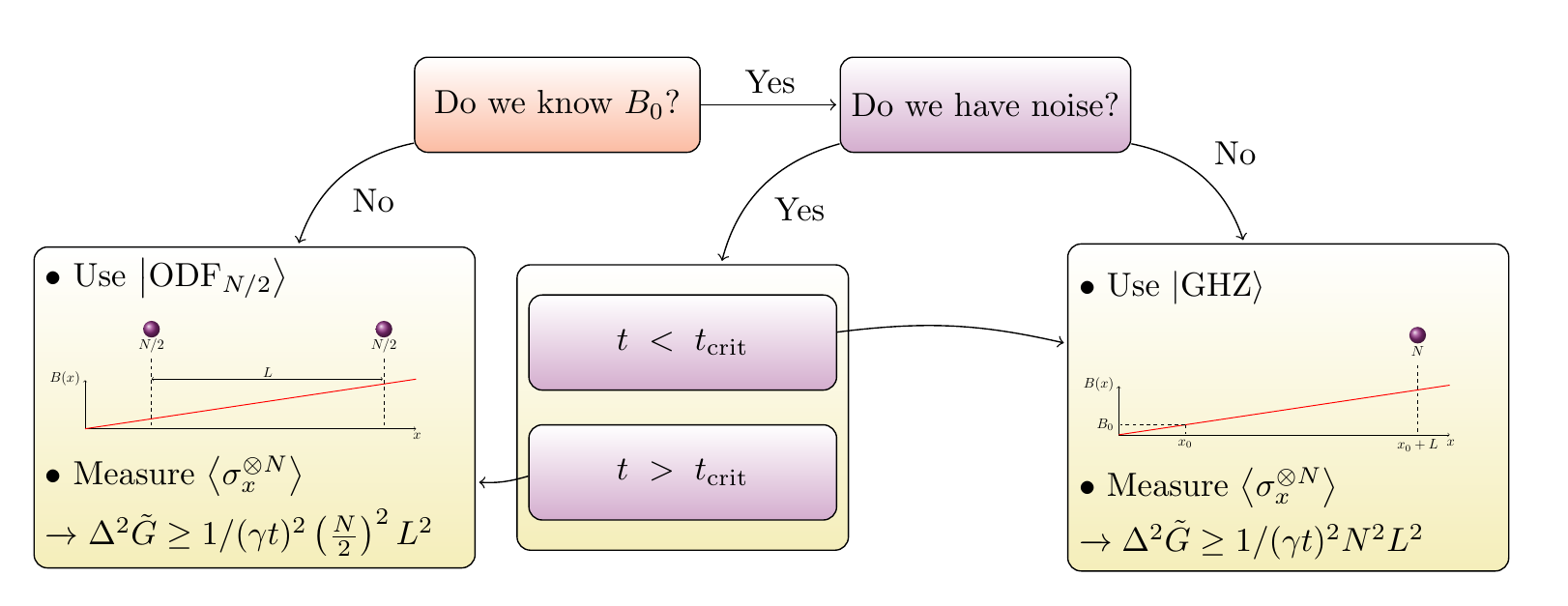}
\caption{
Decision diagram describing how to perform the best gradient measurement. First,  we distinguish between the case of full \textit{a priori} and the case of no \textit{a priori} knowledge about the offset field $B_0$.
When no noise is present, we derived the optimal state, optimal positioning and optimal measurements for both cases. Furthermore, we derived the ultimate limits in precision for both cases.
For the case of full \textit{a priori} knowledge about the offset field $B_0$ we investigated the influence of collective phase noise to the precision bounds and found that this noise source can be interpreted as an erasure of information about $B_0$. Therefore, from the set of investigated states, the GHZ state performs best for $t<t_{\mathrm{crit}}$ and the ODF state with $k=N/2$ performs best for $t> t_{\mathrm{crit}}$.
}\label{fig:flowchart}
\end{figure*}

\begin{table*}
\begin{tabular}{|M{2.5cm}|M{2.5cm}|M{4.8cm}|M{2cm}|M{3.5cm}|N}
\hline
&& general & optimal positioning & equidistant spacing\\
\hline
\multirow{2}{*}{$B_0$ is known} & $\ket \GHZ$ & $ (\gamma t)^2\left[\sum_{i=1}^{N}(x_i-x_0)\right]^{2}$&$(\gamma t)^2  L^2 N^2$ & $(\gamma t)^2 L^{2}\frac{N^{2}}{4}$&\\[20pt]
\cline{2-5}
 & $\ket{\mathrm P}$ & $(\gamma t)^2\sum_{i=1}^{N}(x_i-x_0)^{2}$&$(\gamma t)^2 L^2 N$ &$(\gamma t)^2 L^2  \frac{N(2N-1)}{6(N-1)}$ &\\[20pt]
\hline
\multirow{2}{*}{$B_0$ is not known} & $\ket{\ODF_{N/2}}$ & $(\gamma t)^2\left[\sum_{i=1}^{N/2}\left(x_i-x_{N-i+1}\right)\right]^{2}$& $(\gamma t)^2 L^{2}\frac{N^2}{4}$ &$(\gamma t)^2 L^{2}\frac{N^4}{16\left(N-1\right)^{2}} $ & \\[20pt]
\cline{2-5}
& $\ket{\mathrm P}$ &$(\gamma t)^2\left[\sum_{i=1}^N x_i^2-\frac{1}{N} \left( \sum_{i=1}^N x_i\right)^2\right]$ & $ (\gamma t)^2 L^2 \frac{N}{4}$& $(\gamma t)^2 L^2 \frac{N(N+1)}{12(N-1)}$ & \\[20pt]
\hline
\end{tabular}
\caption{QFI for different states and scenarios considered in this work. For the case of full \textit{a priori} knowledge about the offset field $B_0$ GHZ states [Eq.~\eqref{eq:def:GHZ}] and from the set of separable states $\ket{\mathrm P}$ [Eq.~\eqref{eq:def:product}] are optimal.
For the case of no \textit{a priori} knowledge about the offset field $B_0$ ODF states [Eq.~\eqref{eq:optstatesubspace}] with $k=\frac{N}{2}$ are optimal and are here compared to the separable state $\ket{\mathrm P}$. In this table we list the QFI for a fixed positioning (general), for the optimal positioning, and for equidistant spacing. }\label{tab:overview}
\end{table*}
We presented a systematic analysis of the ultimate limits in precision for the estimation of a gradient of a spatially-varying magnetic field in systems of cold atoms and trapped ions. The position degrees of freedom were treated classically and  taken as fixed. We used the framework of quantum  metrology to study two extreme scenarios: (i) the case when the  magnetic offset field is known and (ii) the case, where the magnetic offset field is not \textit{a priori} known.

For the first case (i) we have introduced the bounds in precision for gradient estimation analogous to the standard quantum limit (maximal possible accuracy with separable states) and the Heisenberg limit (maximal possible accuracy with entangled states) known from the usual phase estimation scenario. Moreover, we have identified the optimal probe state, that is a GHZ state [see Eq.~\eqref{eq:def:GHZ}]. It is then optimal to put all qubits as far away as possible from the point $x_0$, where the magnetic offset field is known. This leads to a magnetic field measurement, similar to a magnetic offset field measurement, but at a different place.

For the second case (ii) we found that GHZ states are completely useless ($\QFI=0$) for the estimation of a magnetic field gradient. In the absence of knowledge about $B_0$ effective super-selection rules restrict the class of allowed states to decoherence-free states. We proved that the decoherence-free state given in Eq.~\eqref{eq:optstatesubspace} with $k=N/2$ excitations is optimal and does not depend on the positions of the qubits.  Here, the optimal positioning is, to put half of the qubits at one place and the other half as far away as possible. We also showed that the performance of optimal decoherence-free states is generically comparable to  optimal GHZ states in the case of complete knowledge about $B_0$ -- both scale with $N^2$. Both optimal states can be prepared with high fidelities in experiments with trapped ions up to $N=14$ and cold atoms up to $N=2$.

For both scenarios, we identified the parity measurement in the $x$ basis as the optimal measurement saturating the quantum Cram{\'e}r-Rao bounds for gradient estimation. This measurement  is feasible in experiments considered in this work, as for the positions of the particles can be considered fixed and local measurement of $\sigma_x$ can be easily performed. 

Finally, we investigated the effect of collective phase noise. Collective phase noise can be interpreted as an erasure of knowledge about the magnetic offset field and continuously interpolates between scenario (i) and (ii) for strong noise or rather long probing times.
We found a critical time $t_\mathrm{crit}$ for which the GHZ state performs as good as the ODF state with $k=N/2$ excitations. For $t<t_c$ GHZ states perform better than ODF states and for $t>t_c$ ODF states outperform GHZ states.  These results are summarized in the decision diagram in Fig.~\ref{fig:flowchart}. Values of the QFI for different positioning and different states in the two cases (i) and (ii) are summarized in Table~\ref{tab:overview}.

We derived Cram\'{e}r-Rao bounds for gradient estimation from the QFI and discussed their saturation with FI. Such a saturation implies an unlimited amount of statistics and therefore many repetitions of a measurement.  However, realistic experiments are limited in measurement time and therefore limited in the amount of possible repetitions. In such a scenario, a proper analysis of bounds in precision can be performed in a Bayesian estimation approach. For the standard scheme [as depicted in Fig. \ref{fig:shemes}(a)] it was shown in Ref. \cite{Berry2000} that only the bound $\Delta^2 \varphi \ge \pi^2/N^2$ can be saturated with limited statistics in contrary to the bound $\Delta^2 \varphi \ge 1/N^2$ from the QFI. An investigation of bounds from a Bayesian approach for the estimation of gradients would be interesting for further work.
Moreover, it would be interesting to take the uncertainty in positioning of the qubits 
into account. In fact, independently from our work, an article on
gradient estimation with systems of atoms with probability distributions in position
has appeared \cite{Apellaniz2017}. 
For such a system also weak value measurements could offer an enhancement in precision for the estimation of gradients 
\cite{Aharonov1988}.
Furthermore, in certain setups there is a coupling between the internal and external 
degrees of freedom, i.e. the spin and the position \cite{Mintert2001}. This requires an adaption
of our ideas. Also, the investigation of precision limits and optimal strategies for simultaneous estimation
of many parameters describing a field (i.e., the offset field B 0 , the gradient G, and higher
derivatives), could be interesting (especially in the presence of collective dephasing noise \cite{Dorner2012,Dorner2013}).   A first step in this direction, however without considering the effect of noise has been done in Ref.~\cite{Apellaniz2017}.
Finally, another interesting topic for further studies is the performance of random multiparticle states \cite{RandSTAT} for gradient estimation.

\section*{Acknowledgements}

We are grateful to
Antonio Ac\'{i}n,
Iagoba Apellaniz,
Michael Johanning, 
Janek Ko\l ody\'{n}ski, 
Morgan Mitchell,
Christina Ritz,
and Gael Sent\'{i}s
 for interesting and fruitful discussions.
This work has been supported by the European Research Council (Consolidator Grants TempoQ/683107
and QITBOX), Spanish MINECO (QIBEQI FIS2016-80773-P,  and Severo Ochoa Grant No. SEV-2015-0522), Fundaci\'{o} Privada Cellex, Generalitat de Catalunya (Grant No. SGR 874, 875, and CERCA Programme), the Friedrich-Ebert-Stiftung, the FQXi Fund (Silicon Valley Community Foundation), and the DFG. M.O acknowledges the support of Homing programme of the Foundation for Polish Science co-financed by the European Union under the European Regional Development Fund.

\begin{appendix}
	
\section{Maximal QFI and optimal states for arbitrary position $x_0$}\label{app:maxFI}	

In this part we derive the maximal QFI for gradient estimation for an arbitrary position of $x_0$, in which the offset field is assumed to be  perfectly known. Let us first introduce the auxiliary notation $f_i \coloneqq x_i-x_0$.  Moreover, just like in the main text let us label the qubits in such a way that $f_{i+1} \geq f_i$. By the virtue of Eq.~\eqref{eq:simpl QFI} we can focus on maximizing  $\QFI[\state,H_G ]$, where 
\begin{equation}\label{eq:HGham}
H_G = \sum_{i=1}^N f_i \frac{\sigma^{(i)}_z}{2}.
\end{equation}
Let us note that the above Hamiltonian can be diagonalized by the computational basis 
\begin{equation}
\ket I=\ket{i_{1}}\ket{i_{2}}\ldots\ket{i_{n}}\ , \ \text{with}\ i_{k}\in\left\{ 0,1\right\} \ . 
 \end{equation}
The symbol $I$  labels the set of positions of particles which are in the state $\ket{0}$ and is formally defined by
\begin{equation}
I=\left\{ l \left|\, i_{l}=0\right.\right\}\ .
\end{equation}
Note that $I$   can be arbitrary (in particular it can be also the empty set). The eigenvalue corresponding to the eigenvector $\ket I$ is given by 
\begin{equation}
\lambda_{I}=\frac{1}{2}\sum_{i\in I}f_{i}-\frac{1}{2}\sum_{i\in\bar{I}}f_{i}\,,
\end{equation}
 where $\bar{I}$  denotes the complement of the set $I$ in the set $\lbrace 1,2,\ldots,N\rbrace$.  The maximal eigenvalue $\lambda_{\max}$ is given by
\begin{equation}
\lambda_{\max} = \frac{1}{2} \sum_{i=1}^N |f_i|
\end{equation} 
with the corresponding eigenvector that is $\ket{I_{\max}}$, where
\begin{equation}
I_{\max} = \left\{ l\left| f_l\geq 0 \right.\right\}\,.
\end{equation}
Let $m$ be the minimal number with $f_m \leq 0$, that is the number of particles on the left side of $x_0$ such that $x_m\le x_0 < x_{m+1} $. Now, using the ordering $f_{i+1}\geq f_i$ we find 
\begin{equation}
\ket{I_{\max}}=\ket{1}^{\ot m} \ot \ket{0}^{\ot N-m}\, .
\end{equation}
Because of $\lambda_{\bar{I}}= - \lambda_{I}$ we get $\lambda_{\min}= -\lambda_{\max}$ and
\begin{equation}
\ket{I_{\min}}=\ket{0}^{\ot m} \ot \ket{1}^{\ot N-m}\, .
\end{equation}
Finally, by virtue of Eq.~\eqref{eq:simpl QFI}, Eq.  \eqref{eq:maxfisch}, and Eq.~\eqref{eq:optSTATE}, we obtain
\begin{equation}
\max_{\state\in\D(\H_N)} \QFI (\state)=(\gamma t)^2 \left[\sum_{i=1}^N |f_i|\right]^2\, ,
\end{equation}
with the optimal state given by
\begin{equation}
\ket{\Psi_m}=\frac{1}{\sqrt{2}}\left(\ket{1}^{\ot m} \ot \ket{0}^{\ot N-m}  +\ket{0}^{\ot m} \ot \ket{1}^{\ot N-m}\right)\ .
\end{equation}
Note that for $m=N/2$ this state happens to be the optimal decoherence-free state $\ket{\ODF_{N/2}}$.
	
\section{Parity measurements}\label{app:FI}
 
In this part of the Appendix we compute expectation values of the parity operator $\hat{M}=\sigma_{x}^{\ot N}$ on the families of quantum states investigated in this paper. These computations are relevant for the computations involving the classical Fisher information and the error propagation formula given the main text.
 
\subsection{Parity expectation value for GHZ states}\label{subsub:GHZ}
Recall that $\ket{\GHZ}=\frac{1}{\sqrt{2}}( \ket{0}^{\ot N} + \ket{1}^{\ot N})$. Using the identities
\begin{equation}
U_G \ket{0}^{\ot N}\!=\! \exp\!\!\left[-\frac{\ii}{2}\!\left(N \gamma B_0 t+\gamma G t \sum_{i=1}^N (x_i -x_0)\!\!\right)\!\right] \ket{0}^{\ot N},
\end{equation}
\begin{equation}
U_G \ket{1}^{\ot N}\!=\! \exp\!\!\left[\frac{\ii}{2}\!\left(N \gamma B_0 t+\gamma G t \sum_{i=1}^N (x_i -x_0)\!\!\right)\!\right]  \ket{1}^{\ot N},
\end{equation}
and the property $\hat{M} \ket{0}^{\ot N}= \ket{1}^{\ot N}$ we obtain 
\begin{equation} \label{eq:ParityGHZ}
\tr\left(U_G \psi_{\GHZ} U^\dagger_G  \hat{M} \right)=\!\cos\!\!\left[ N \gamma B_0 t+\!\gamma t G \sum_{i=1}^N (x_i-x_0)\! \right]. 
\end{equation}
\subsection{Parity expectation value for noisy GHZ states}\label{app:Parity_noisyGHZ}
Let $\psi_{\GHZ,\theta}\coloneqq\kb{\GHZ_\theta}{\GHZ_\theta}$, where 
\begin{equation}
\ket{\GHZ_\theta}=\frac{1}{\sqrt{2}}\left(\ket{0}^{\ot N}+ \exp(\ii \theta) \ket{1}^{\ot N}  \right) \ .
\end{equation}
In this part we compute the expectation value of the parity $\sigma_{x}^{\ot N}$ on the noisy state $\rho:=U_G \bar{\psi}_{\GHZ,\theta}(t) U^\dagger_G$.  Using Eq.~\eqref{eq:evolvedGHZ} we find
\begin{align}
\begin{split}
\rho&=d(t) U_G \psi_{\GHZ,\theta} U^\dagger_G \\
&+ \left(1-d(t)\right)\frac{1}{2}\left[ (\kb{0}{0})^{\ot N}  + (\kb{1}{1})^{\ot N}\right] \ , 
\end{split}
\end{align}
where $d(t)$ is given below Eq.~\eqref{eq:evolvedGHZ} in Appendix~\ref{app:GHx_under_noise}. From the above expression and using $\hat{M}\ket{0}^{\ot N}= \ket{1}^{\ot N}$, we get 
\begin{equation}
\tr\left(\rho \hat{M}\right)=d(t)\tr\left(U_G \psi_{\GHZ,\theta} U^\dagger_G \hat{M}\right)\ . 
\end{equation}
Finally, repeating essentially the same computations as the ones from Section~\ref{subsub:GHZ} we obtain 
\begin{equation}\label{eq:ParityGHZnoise}
\tr\left(\rho \hat{M}\right)=d(t)\cos\left[ N \gamma B_0 t+ \gamma t G \sum_{i=1}^N (x_i-x_0)+\theta \right]\ .
\end{equation}
\subsection{Parity expectation value for optimal decoherence-free states}
Repeating the analogous computations to these given in Appendix~\ref{subsub:GHZ}  for 
\begin{equation}
\ket{\ODF_{N/2}}=\frac{1}{\sqrt{2}}\left(\ket{1}^{\otimes N/2}\otimes \ket{0}^{\ot N/2} +\ket{0}^{\ot N/2}  \otimes \ket{1}^{\otimes N/2}  \right)
\end{equation}
we obtain
\begin{equation}
U_G \ket{a}= \exp{\left[-\frac{\ii}{2}\left(\gamma t G \sum_{i=1}^{N/2}\left(x_i-x_{N-i+1}\right) \right)\right]} \ket{a}\ ,
\end{equation}
\begin{equation}
U_G \ket{b}= \exp\left[\frac{\ii}{2}\left(\gamma t  \sum_{i=1}^{N/2}\left(x_i-x_{N-i+1}\right) \right)\right] \ket{b} \ ,
\end{equation}
where we denoted $\ket{a}:=\ket{1}^{\otimes N/2}\otimes \ket{0}^{\ot N/2}$ and $\ket{b}:=\ket{0}^{\ot N/2} \ot \ket{1}^{\otimes N/2}$.
 Using the above expressions together with the identity $\hat{M}\ket{b} =\ket{a}$ we obtain
 
\begin{equation}
\tr\left(U_G \psi^{N/2}_{\ODF} U^\dagger_G  \hat{M} \right)=\cos\left(\gamma t \sum_{i=1}^{N/2}\left(x_i-x_{N-i+1}\right) G \right) \ ,\label{eq:parity_expect_ODF}
\end{equation}
where $\psi^{N/2}_{\ODF}=\kb{\ODF_{N/2}}{\ODF_{N/2}}$.
 
\section{Classical Fisher information for $J_x$ measurement}\label{app:JXmes}
Analogous computations to the ones performed in Section~\ref{sec:full_a_priori_knowledge} can be performed to show that the classical Fisher Information for the measurement of the projective POVM $\PPP_{J_x}$, associated with the eigenspaces of $J_x$, also gives the QFI for the optimal states. More precisely
\begin{equation}\label{eq:EqualityFI}
\FI\left(\state_G, \PPP_{J_x}\right)= \FI\left(\state_G, \PPP \right)\ ,
\end{equation}
for $\state = \psi_{\mathrm{GHZ}}$ and  $\state=\psi^{N/2}_{\ODF}$. 

This result can be also derived from the monotonicity of QFI under coarse-graining, i.e., 
\begin{equation}\label{eq:monotonicityFI}
\FI\left(\state_G, \lbrace M_i\rbrace\right)\geq  \FI\left(\state_G, \lbrace N_i\rbrace\right)\ ,
\end{equation}
where $\lbrace N_i \rbrace$ is a POVM obtained by coarse-graining of a POVM  $\lbrace M_i \rbrace$ i.e for every outcome $i$ we have  $N_i =\sum_{j} q(i|j) M_j$ for some stochastic matrix $q(i|j)$  (we call a matrix $q(i|j)$ stochastic if and only if $q(i|j)\geq 0$ and for every $j$ we have $\sum_i q(i|j)=1$). A measure $\PPP$, describing the measurement of the parity in $x$ basis $\sigma_x ^{\otimes N}$, can be obtained as follows: First, measuring the projective measurement $\PPP_{J_x}$. Second, output "+1" or "-1" depending on the number of excitations contributing to the observed eigenvalue of $J_x$. Therefore, $\PPP$ is coarse-graining of  $\PPP_{J_{x}}$ and thus, by the virtue of Eq.~\eqref{eq:monotonicityFI}  we obtain Eq.~\eqref{eq:EqualityFI} for arbitrary states $\state$.

\section{Computations of error-propagation formula for GHZ states in presence of noise}\label{app:EPF_noisyGHZ}

The aim of this section is to show that for a suitable chosen value of the initial relative phase in a state $\ket{\GHZ_\theta}$ it is possible to saturate the quantum Cram\'{e}r-Rao bound with the measurement of $\hat{M}=\sigma_{x}^{\ot N}$, even in the presence of collective phase noise. Our reasoning essentially mimics the one given in the previous section. Setting $\psi_{G,\mathrm{noise}}=U_G \bar{\psi}_{\GHZ,\theta} U^{\dagger}_G$ and using  \eqref{eq:ParityGHZnoise} we obtain 
\begin{equation}
\Delta_{\psi_{G,\mathrm{noise}}}^2 \hat{M}=1-d(t)^2\cos^2\left[ \alpha(t)\right] \ ,
\end{equation}
with $\alpha(t):= N\gamma B_0 t+\gamma G t \sum_{i=1}^N (x_i-x_0)+\theta$.
Using this formula in the error propagation formula given in Eq.~\eqref{eq:error_propagation_formula} we get 
\begin{align}
\begin{split}
\Delta^2 \tilde{G}_{\hat{M}}&=\frac{1-d(t)^2\cos^2\left[\alpha(t)\right]}{\left[d(t)\gamma t \sum_{i=1}^N (x_i-x_0) \right]^2 \sin^2\left[\alpha(t)\right]}\, ,\\
&=\frac{1+\left[1-d(t)^2\right] \mathrm{cot}^2\left[\alpha(t)\right]}{\left[d(t)\gamma t \sum_{i=1}^N (x_i-x_0) \right]^2 } \, .
\end{split}
\end{align}
From this we see that the quantum Cram\'{e}r-Rao bound [in this case given by the inverse of the QFI in Eq.  \eqref{eq:QFI_noisy_GHZ}] is saturated for $\mathrm{cot}\left[\alpha(t)\right]=0$ (for a suitable choice of the initial phase $\theta$).

\section{Optimal states in presence of collective phase noise} \label{app:optimal_under_noise}
In this section we will assume full \textit{ a priori} knowledge about $B_0$. We calculate the QFI first for GHZ states in presence of noise and then for optimal states for arbitrary position $x_0$  in presence of noise.
\subsection{GHZ states in presence of collective phase noise} \label{app:GHx_under_noise}
Due to the noise, the probe state $\varrho$ evolves into a mixed state $\bar{\state}(t)\coloneqq \langle U_{\mathrm{noise}} \state U^\dagger_{\mathrm{noise}}\rangle$ where $U_\mathrm{noise}=\mathrm{exp}\left[-\ii\gamma' \int_0^t\Delta E(\tau) \mathrm{ d}\tau J_z\right]$ describes the noise acting on the system.
The diagonal entries of the probe state do not change $\bar{\varrho}(t)_{ii}=\varrho_{ii}$. However, the off-diagonal ones do. The GHZ state has only two non-zero off-diagonal entries $\varrho_{0,q}=(\varrho_{q,0})^\dagger$, where $q=2^N-1$ for a state of dimension $2^N$. For these entries we find \cite{Monz2011}
\begin{align}
\begin{split}
\left[U_{\mathrm{noise}} \state U^\dagger_{\mathrm{noise}}\right]_{0,q}&= \mathrm{exp}\left[-\ii\gamma' N \int_0^t\Delta E(\tau) \mathrm{ d}\tau \right]\state_{0,q}\, .
\end{split}
\end{align}
Now we use the fact that $\langle \exp[\pm\ii\delta\varphi]\rangle=\exp[-\frac{1}{2}\Delta^2\delta\varphi]$ for an unbiased Gaussian distribution of $\delta\varphi$ that means that $\braket{\delta \varphi}=0$.
With the fact that the variance $\Delta^2\delta\varphi=\braket{\delta \varphi^2}$ we can calculate 
 \begin{align}
 \left\langle \exp\left[\pm\ii\gamma' N \int_0^t\Delta E(\tau) \mathrm{ d}\tau \right]\right\rangle=\exp\left[-\frac{1}{2}(\gamma' N)^2 C(t)\right]
 \end{align}
with $C(t)=\langle\int_0^t\mathrm{ d}\tau_1\int_0^t\mathrm{ d}\tau_2\Delta E(\tau_1) \Delta E(\tau_2) \rangle$.
Substituting $t_1=\tau_1-\tau_2$ and $t_2=\tau_1+\tau_2$ and using $\braket{E(t+\tau)E(t)}=\braket{E(\tau)E(0)}$ and $\braket{\Delta E(t) \Delta E(0)}=(\Delta E)^2 \exp\left[-t/\tau_c \right]$ we find 
\begin{align}
C(t)=2(\Delta E \tau_c)^2 \left(\exp(-t/\tau_c)+t/\tau_c -1\right)\,.
\end{align}
Together, we find the $N$-particle GHZ state evolves in presence of collective phase noise into the state 
\begin{align}\label{eq:evolvedGHZ}
\begin{split}
\bar{\varrho}(t)&=\frac{1}{2} \kb{0^{\otimes N}}{0^{\otimes N}}+\frac{1}{2} \kb{1^{\otimes N}}{1^{\otimes N}}\\&+ \frac{d(t)}{2} \kb{0^{\otimes N}}{1^{\otimes N}}+ \frac{d(t)}{2}\kb{1^{\otimes N}}{0^{\otimes N}}
\end{split}
\end{align}
with $d(t)=\mathrm{exp}\left[-\left(N\gamma' \Delta E \tau_c\right)^2 \left(\exp(-t/\tau_c)+t/\tau_c -1\right)\right]$. 
This state in its eigendecomposition is given by the non-zero eigenvalues $\lambda_{\pm}=\frac{1}{2}(1 \pm d(t))$ and the corresponding eigenvectors $\ket{v_{\pm}}=\frac{1}{\sqrt{2}}\left(\ket{0 \cdots 0} \pm \ket{1 \cdots 1}\right)$. 
In order to compute the QFI for a noisy GHZ state and for estimating $G$ we use Eq.~\eqref{eq:simpl QFI} and Eq.~\eqref{eq:QFI}, with the final result 
\begin{align}
\begin{split}
\QFI(\bar{\state}(t))&= \frac{(\lambda_+-\lambda_-)^2}{\lambda_++\lambda_-}(\gamma t)^2 |\bk{v_+|\sum_{i=1}^N (x_i -x_0) \sigma_z^{(i)}}{v_-}|^2\\&=  d(t)^2 (\gamma t)^2 \left[\sum_i (x_i-x_0)\right]^2 \ ,
\end{split}
\end{align}
where all other terms in Eq.~\eqref{eq:QFI} vanish since 
\begin{equation}
\sum_{i=1}^N (x_i-x_0) \sigma_z^{(i)} \ket{v_+}=\sum_{i=1}^N (x_i-x_0) \ket{v_-} \ .
\end{equation} 
 
\subsection{Optimal states for arbitrary position $x_0$ in presence of noise}\label{app:opt_state_under_noise}
Let $m$ be the minimal number with $f_m \leq 0$, that is the number of particles on the left side of $x_0$.
Then, in Appendix  \ref{app:maxFI} we showed that the optimal state is given by
\begin{equation}
\ket{\Psi_m}=\frac{1}{\sqrt{2}}\left(\ket{1}^{\ot m} \ot \ket{0}^{\ot N-m}  +\ket{0}^{\ot m} \ot \ket{1}^{\ot N-m}\right)\ .
\end{equation}
Following the calculations from the previous section we find the averaged state for a given time $t$
\begin{align}
\begin{split}
\bar{\varrho}(t)&=\frac{1}{2} \kb{1^{\otimes m},0^{\otimes N-m}}{1^{\otimes m},0^{\otimes N-m}}\\
&+\frac{1}{2} \kb{0^{\otimes m},1^{\otimes N-m}}{0^{\otimes m},1^{\otimes N-m}}\\
&+\frac{d_m(t)}{2} \kb{1^{\otimes m},0^{\otimes N-m}}{0^{\otimes m},1^{\otimes N-m}}\\
&+\frac{d_m(t)}{2} \kb{0^{\otimes m},1^{\otimes N-m}}{1^{\otimes m},0^{\otimes N-m}},
\end{split}
\end{align}
with 
\begin{align}
d_m(t):=\mathrm{exp}\left[-(N-2m)^2\left(\gamma'\Delta E \tau_c\right)^2 \left(\!\mathrm{e}^{- \frac{t}{\tau_c}}+\frac{t}{\tau_c} -1\right)\right].
\end{align}
In the cases of $m=0$ and $m=N$ the optimal state is a GHZ state $\ket{\Psi_0}=\ket{\Psi_N}=\ket{\GHZ}$ and we find $d_0(t)=d_N(t)=d(t)$.
The non-zero eigenvalues of $\bar{\varrho}(t)$ are given by $\lambda^m_\pm (t)=\frac{1\pm d(m,t)}{2}$ with the corresponding eigenvectors
\begin{equation}
\ket{v_\pm^m}=\frac{1}{\sqrt{2}}\left(\ket{1}^{\ot m} \ot \ket{0}^{\ot N-m} \pm\ket{0}^{\ot m} \ot \ket{1}^{\ot N-m}\right).
\end{equation}
Now we can use the fact that 
\begin{equation}
\sum_{i=1}^N (x_i-x_0) \sigma_z^{(i)} \ket{v^m_+}=\left(\sum_{i=1}^N |x_i-x_0| \right)\ket{v^m_-} \ .
\end{equation} 
 to evaluate the QFI for estimating $G$ that is given by
\begin{align}
\begin{split}
F_Q&= \frac{(\lambda^m_+ -\lambda^m_-)^2}{\lambda^m_+ +\lambda^m_-}(\gamma t)^2 |\bk{v_+^m|\sum_{i=1}^N (x_i-x_0) \sigma_z^{(i)}}{v_-^m}|^2\,\\
&=(\gamma t)^2 d_m(t)^2 \left(\sum_{i=1}^N |x_i-x_0| \right)^2.
\end{split}
\end{align}

\section{Optimal product state in the steady state regime}\label{app:product_state_steady}
In the noiseless case, the product state $\ket{\mathrm P}$  is the best classical probe state. We now want to understand, what noise (loosing information about $B_0$) does to the scaling for this state.
The state $\bar{\psi}_P (t\rightarrow\infty)=\sum_{k=0}^N p_k \ket{\D_N^k}\bra{\D_N^k}$ is a mixture of symmetric Dicke states $\ket{\D_N^k}$ \cite{Dicke1954} with probabilities $p_k=2^{-N} \binom{N}{k}$, where $\sum_{k=0}^N p_k=1$. Recall first that 
\begin{equation}
H_G =\sum_{i=1}^N f_i  \frac{\sigma_z^{(i)}}{2}\ ,
\end{equation}
where we set for convenience $f_i\coloneqq x_i - x_0$.   We perform the calculations analogous to the ones given in Ref.~\cite{Altenburg2016}.  Because of the fact that $H_G$ preserves subspaces $\V_k$ and $\ket{\D_N^k} \in \V_k$ we have $\bk{\D_N^s|H_G }{\D_N^l}\propto \delta_{l,s}$
and therefore the QFI reduces to
\begin{align}
\QFI =4 (\gamma t)^2 \sum_{k=0}^N p_k \Delta_k^2 H_G,
\end{align}
with $\Delta_k^2 H_G$ being the variance of $H_G$ in the state $\ket{\D_N^k}$.
The second term of the variance is the squared expectation value $\braket{H_G}$, which is given by
\begin{align}
\begin{split}
\bk{\D_N^k|\sum_{i=1}^N f_i  \frac{\sigma_z^{(i)}}{2}}{\D_N^k}&=\sum_{i=1}^N f_i\bk{\D_N^k|  \frac{\sigma_z^{(i)}}{2}}{\D_N^k}\\&=\sum_{i=1}^N f_i \frac{1}{N}\bk{\D_N^k| J_z}{\D_N^k}\\&=\sum_{i=1}^N f_i\frac{ \left(2k-N\right)}{2N},
\end{split}
\end{align}
using the symmetry of the state.
The expectation value of the squared operator $\braket{H_G^2}$ is
\begin{align}
\begin{split}
\bk{\D_N^k|\left[\sum_{i=1}^N f_i  \frac{\sigma_z^{(i)}}{2}\right]^2}{\D_N^k}&=\sum_{i,j=1}^{N} f_i f_j \bk{\D_N^k| \frac{\sigma_z^{(i)}}{2} \frac{\sigma_z^{(j)}}{2}}{\D_N^k}\\&=\sum_{i=1}^N f_i^2\underbrace{\bk{\D_N^k|\left( \frac{\sigma_z^{(i)}}{2}\right)^2}{\D_N^k}}_{=1/4}\\&+ \sum_{i \neq j=1}^N f_if_j\bk{\D_N^k| \frac{\sigma_z^{(i)}}{2} \frac{\sigma_z^{(j)}}{2}}{\D_N^k}.
\end{split}
\end{align}
Using the symmetry of the state  $\bk{\D_N^k|\sigma_z^{(i)}\sigma_z^{(j)}}{\D_N^k}=\bk{\D_N^k|\sigma_z^{(a)}\sigma_z^{(b)}}{\D_N^k}$ for arbitrary $a$ and $b$, we can rewrite the second term
\begin{align}
\begin{split}
\bk{\D_N^k| \frac{\sigma_z^{(i)}}{2} \frac{\sigma_z^{(j)}}{2}}{\D_N^k}&=\frac{1}{ N(N-1)} \sum_{a \neq b=1}^N \bk{\D_N^k| \frac{\sigma_z^{(a)}}{2} \frac{\sigma_z^{(b)}}{2}}{\D_N^k}\\&=\frac{1}{N(N-1)} \left(\underbrace{\bk{\D_N^k|J_z^2}{\D_N^k}}_{=(2k-N)^2/4}-\frac{N}{4}\right).
\end{split}
\end{align}
Together the variance is given by
\begin{align}
\begin{split}
4\Delta_k^2 H_G &= \sum_{i=1}^N f_i^2-\left( \sum_{i=1}^N f_i\right)^2 \frac{\left(2k-N\right)^2}{ N^2}\\&+\sum_{i \neq j=1}^N f_if_j\left[\frac{(2k-N)^2-N}{N(N-1)}\right].\label{app:eq:QFI_Dicke}
\end{split}
\end{align}
Here, all terms with $x_0$ vanish. Therefore, $4\Delta_k^2 H_G$ is independent of $x_0$.
Using $\sum_{k=0}^N 2^{-N}\binom{N}{k} \left(2k-N\right)^2=N$, we can calculate the QFI
\begin{align}
\begin{split}
\QFI&= 4 (\gamma t)^2 \sum_{k=0}^N p_k \Delta_k^2 H_G\\&=(\gamma t)^2\left[\sum_i x_i^2-\frac{1}{N} \left( \sum_i x_i\right)^2\right].\label{app:eq:QFI_prod_steady}
\end{split}
\end{align}

For the maximization of Eq.~\eqref{app:eq:QFI_prod_steady} over the positioning $\{x_i\}$ we can state:
\begin{lem}
Let $N$ be an even natural number and let
\begin{equation}
f\left(x_1,x_2,\ldots,x_N\right)\coloneqq \sum_{i=1}^N x_i^2-\frac{1}{N} \left( \sum_{i=1}^N x_i\right)^2 \ .
\end{equation} 
Then, the restriction of $f$ to the domain $x_i \in[x_0, x_0 +L]$ attains the maximum value for $x_i=x_0$, for $i=1,\ldots,N/2$ and $x_i=x_0 +L$ for $i=N/2+1.\ldots,N$.
\end{lem}
\begin{proof}
A direct computation shows that for any $\delta\in\mathbb{R}$ we have
\begin{equation}
f\left(x_1+\delta,x_2+\delta,\ldots,x_N+\delta\right)=f\left(x_1,x_2,\ldots,x_N\right) \ .
\end{equation}
Therefore, the problem of maximizing $f$ in the domain specified by the restrictions $x_i \in[x_0, x_0 +L]$ can be reduced to the problem of maximizing this function for its arguments satisfying $x_i \in[-L/2, L/2]$. For such  restrictions setting half of $x_i$ equal to $-L/2$ and the other half equal to $L/2$ maximizes $\sum_{i=1}^N x_i^2$ while at the same time minimizing $\left( \sum_{i=1}^N x_i\right)^2$. Thus, such configuration maximizes $f$ for $x_i \in[-L/2, L/2]$. Coming back to the original interval we get the thesis of the lemma.
\end{proof}

\section{Bounds in precision for gradient estimation with no \textit{a priori} knowledge about $B_0$} \label{app:presicion_bounds_not_knowing_b}

In this Section we derive the bounds in precision for gradient estimation under the assumption of no \textit{a priori} knowledge about the offset field $B_0$. In particular, we will prove Eq.~\eqref{eq:sequentialreduction}, Eq.~\eqref{eq:maxfischsubspace}, Eq.~\eqref{eq:optstatesubspace}, and Eq.~\eqref{eq:optimalDFS}  given in Section~\ref{sec:no_a_priori_knowledge}.

Let us start with proving Eq.  \eqref{eq:sequentialreduction} which reads 
\begin{equation}\label{eq:sequentialreduction2}
\max_{\state \in \DFS_N} \QFI(\state) = (\gamma t)^2 \max_{k=0,\ldots,N}\, \max_{\state\in\D(\mathcal{V}_k)} \QFI[\state,H_G] \,.
\end{equation}
In order to prove this equation we first recall the connection $\QFI(\state)=(\gamma t)^2 \QFI[\state,H_G]$. Then, for the "Hamiltonian" quantum Fisher information we have
\begin{equation}\label{eq:sumQFI}
\QFI\left[\sum_{k=0}^N p_k\state_k,H_G\right]=  \sum_{k=0}^N p_k  \QFI[\state_k,H_G]\ ,
\end{equation}  
where $\lbrace p_k \rbrace$ is a probability distribution and states $\state_k$ are supported on $\V_k$. Eq.~\eqref{eq:sumQFI} follows from the fact that $H_G$ preserves decoherence-free subspaces $\V_k$ and the "additivity" QFI under the convex combinations of states supported on orthogonal subspaces \cite{Toth2014}. The identity in Eq.~\eqref{eq:sequentialreduction2} follows now from the linearity of the right-hand side of Eq.~\eqref{eq:sumQFI} in $\lbrace p_k\rbrace$ and the fact that decoherence-free states are precisely of the form $\sum_{k=0}^N p_k\state_k$ for $\state_k$ supported on $\V_k$.

The optimal value of the QFI on $\D(\V_k)$ can be found by using   Eq.~\eqref{eq:maxfisch} and Eq.~\eqref{eq:optSTATE}. Let $\lambda_{\max}^{(k)}$ and $\lambda_{\min}^{(k)}$ denote the maximal and respectively minimal eigenvalues of $\left.H_G\right|_{\mathcal{V}_{k}}$ [the formula for $H_G$ is given for example in Eq.~\eqref{eq:HGham}]. Using the monotonicity of the coefficients $f_{i+1}\ge f_{i}$ with $f_i=(x_i-x_0)$ we get
\begin{align}
\lambda_{\max}^{\left(k\right)} & =\sum_{i=k+1}^{N}f_{i}-\sum_{i=1}^{k}f_{i}\,,\label{eq:restricted optimal}\\
\lambda_{\min}^{(k)} & =\sum_{i=1}^{N-k}f_{i}-\sum_{i=N-k+1}^{N}f_{i}\,.\label{eq:resticted minimal}
\end{align}
The corresponding eigenvectors are given by 
\begin{align}
\ket{I_{\max}^{(k)}} & =\left(\ket 1^{\otimes k}\right)\otimes\left(\ket 0^{\otimes N-k}\right)\,,\label{eq:explicitformvec}\\
\ket{I_{\min}^{(k)}} & =\left(\ket 0^{\otimes N-k}\right)\otimes\left(\ket 1^{\otimes k}\right)\,.\nonumber 
\end{align}
Using Eq.~\eqref{eq:restricted optimal} and Eq.~\eqref{eq:resticted minimal}
we get 
\begin{equation}
\lambda_{\max}^{(k)}-\lambda_{\min}^{(k)}=\sum_{i=1}^{l}\left(f_i-f_{N-i+1}\right)\,,\,l=\min\{k,N-k\}\,.\label{eq:partial explicit}
\end{equation}
From Eq.~\eqref{eq:simpl QFI} and Eq.~\eqref{eq:partial explicit} we obtain the explicit formula for
the maximal QFI on $\mathcal{V}_{k}$,
\begin{equation}\label{app:eq:maxfischsubspace}
\max_{\state\in \D(\V_k)}\QFI(\state) =(\gamma t)^2\left[\sum_{i=1}^{l}\left(f_i-f_{N-i+1}\right)\right]^{2}\ ,
\end{equation}
where $l=\min\{k,N-k\}$. From Eq.~\eqref{eq:optSTATE} we find that the above value is attained for the state 
\begin{equation}\label{app:eq:optstatesubspace}
\ket{\ODF_k}=\frac{1}{\sqrt{2}}\left(\ket{I_{\max}^{(k)}}+\ket{I_{\min}^{(k)}}\right) 
\end{equation}
We have therefore proved Eq.~\eqref{eq:maxfischsubspace} and Eq.~\eqref{eq:optstatesubspace}. 

We conclude by noting that the right hand side of Eq.~\eqref{app:eq:maxfischsubspace} is a monotonic function in $k$ for $2k\leq N$ and $\max_{\state\in \D(\V_k)}\QFI(\state) = \max_{\state\in \D(\V_{N-k})}\QFI(\state)$ -- see Fig.~\ref{fig:particles_sums} for a graphical explanation of this fact. Therefore, the QFI is maximal for $k=\left\lfloor \frac{N}{2}\right\rfloor$, where $\lfloor n \rfloor$ is the smallest integer smaller or equal to $n$ that is called the floor of $n$. Using Eq.~\eqref{eq:sequentialreduction2}  we get
\begin{equation}
\max_{\state \in \DFS_N} \QFI(\state)=(\gamma t)^2\left[\sum_{i=1}^{\left\lfloor \frac{N}{2}\right\rfloor }\left(x_i-x_{N-i+1}\right)\right]^{2}\,,\label{app:eq:max Fischer}
\end{equation}
This
maximum is obtained for the state 
\begin{equation}
\ket{\ODF_{\left\lfloor \frac{N}{2}\right\rfloor}}=\frac{1}{\sqrt{2}}\left[\ket 1^{\otimes \left\lfloor \frac{N}{2}\right\rfloor }\ket 0^{\otimes \left\lceil \frac{N}{2}\right\rceil }+\ket 0^{\otimes \left\lceil \frac{N}{2}\right\rceil }\ket 1^{\otimes \left\lfloor \frac{N}{2}\right\rfloor }\right]\,,\label{eq:optimaldecohstatenice}
\end{equation}
where $\lceil n \rceil$ is the highest integer greater or equal to $n$ that is called the ceil of $n$.
Let us note that if $N$ is not even also the state 

\begin{equation}
\ket{\ODF_{\left\lceil \frac{N}{2}\right\rceil}}=\frac{1}{\sqrt{2}}\left(\ket{I_{\max}^{\left\lceil \frac{N}{2}\right\rceil}}+\ket{I_{\min}^{\left\lceil \frac{N}{2}\right\rceil}}\right)\label{eq:alternative opt decoh}
\end{equation}
attains the maximal QFI given in Eq.~\eqref{app:eq:max Fischer}.
\begin{figure}
\begin{center}
	\includegraphics[width=0.4\textwidth]{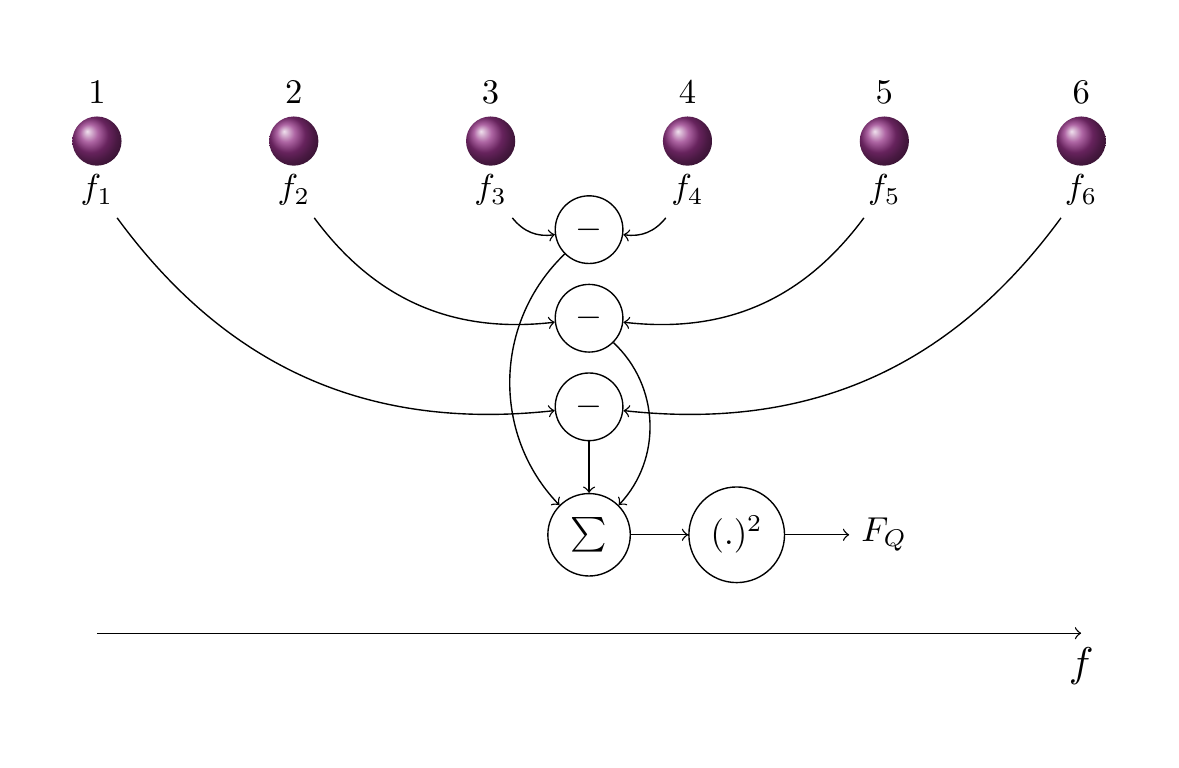}
\end{center}
	\caption{Graphical illustration of Eq.~\eqref{app:eq:maxfischsubspace} for $N=6$ qubits.
	First, $f_{N-i+1}$ is subtracted from $f_i$. These differences will be summed up. Then, the QFI is given by the square of the sum.
	For $k\le N/2$ the number of terms in the summation increases with $k$ and decreases with $k$ for $k>N/2$. This is due to the fact, that the summation cutoff is given by $l=\min \{k,N-k\}$. In total, $\max_{\state\in \D(\V_k)}\QFI(\state)$ is a monotonic function of $k$ for $k\leq N/2$. 
	}\label{fig:particles_sums}
\end{figure}

\subsection{Optimal positioning of the qubits}
Consider $N$ particles that are set to be located in the interval
$x_i\in\left[\tilde x_0,L+\tilde x_0\right]$. Here, $\tilde x_0$ is an arbitrary point. We are interested in how to optimally locate
the qubits in order to get the best possible accuracy for the estimation of $G$
(for fixed $N$ and $L$). According
to Eq.~\eqref{app:eq:maxfischsubspace} the maximal QFI, attainable with a state from the $\DFS_N$ is given by 
\begin{equation}
\QFI\left(\state\right)=(\gamma t)^2\left[\sum_{i=1}^{\left\lfloor \frac{N}{2}\right\rfloor }\left(x_i-x_{N-i+1}\right)\right]^{2}\,.\label{eq:optimfisheven}
\end{equation}
The maximum of \eqref{eq:optimfisheven} over the locations of all particles
$\left\{ x_{i}\right\} _{i=1}^{N}$  is attained for $x_i=\tilde x_0$ for $i\le \left\lfloor \frac{N}{2}\right\rfloor$ and $x_i=L+\tilde x_0$ for $i> \left\lfloor \frac{N}{2}\right\rfloor$ or vice versa.
Then, the maximal QFI is given by
\begin{equation}
\QFI=(\gamma t)^2\left\lfloor \frac{N}{2}\right\rfloor ^{2}L^{2}\,.\label{eq:max decoh fiisher location2}
\end{equation}

Let us note that in the case when $N$ is odd the position of the
``middle'' particle can be arbitrary. We want to emphasize that
the scaling behavior (with respect to $N$ and $L$) of the (over the choice of $x_{i}$'s)  optimized
QFI is preserved if one picks
the optimal state which is invariant to the considered noise model. 
\subsection{Crosspoint GHZ in presence of noise and optimal state from the DFS}\label{app:critical_time}
In the noiseless case the GHZ state is optimal for gradient estimation when $B_0$ is known. However, collective phase noise causes an erasure of knowledge about the offset field $B_0$. In the limit of no knowledge about $B_0$ we found an optimal state from the $\DFS_N$ given in Eq.~\eqref{app:eq:optstatesubspace} with $k=\lfloor N/2 \rfloor$. In total the maximal attainable QFI for this state is smaller then for the GHZ state in the noiseless case. Therefore, in this section we calculate the measurement time $t_\mathrm{crit}$ in which both perform similar.
The QFI for GHZ states in presence of collective phase noise is given by
\begin{equation}
\QFI=  d(t)^2 (\gamma t)^2 \left[\sum_{i=1}^N (x_i-x_0)\right]^2\, ,
\end{equation}
with $d(t)=\mathrm{exp}\left[-\left(N\gamma' \Delta E \tau_c\right)^2 \left(\exp(-t/\tau_c)+t/\tau_c -1\right)\right]$.
The QFI for the optimal state from the $\DFS_N$ with $k=\lfloor N/2 \rfloor$ is given by
\begin{equation}
\QFI=(\gamma t)^2\left[\sum_{i=1}^{\left\lfloor \frac{N}{2}\right\rfloor }\left(x_{N-i+1}-x_{i}\right)\right]^{2}\,.
\end{equation}
Then, we can calculate the critical time $t_\mathrm{crit}$ by setting both equal and solve for $t$.
In realistic experiments the correlation time $\tau_c \propto\,$s and the measurement time $t\propto\,$ms. Such that we can assume $t/\tau_c \ll 1$ which leads to $[\exp(-t/\tau_c)+t/\tau_c -1]\approx 1/2(t/ \tau_c)^2$ and we find
\begin{equation}
t_\mathrm{crit}=\frac{\left\{ 2\log\left[ \frac{\left(\sum_{i=1}^N (x_i-x_0)\right)^2}{\left(\sum_{i=1}^{N/2}\left(x_i-x_{N-i+1}\right)\right) ^{2}}\right]\right\}^{1/2}}{N\gamma' \Delta E}\,.
\end{equation}

\section{Spatial distributions used in Fig.~\ref{fig:max_qfi_vs_k} }\label{app:spatial_distributions}
In Fig.~\ref{fig:max_qfi_vs_k} we illustrated the QFI with a state from the $\DFS_N$ with $k$ excitations for different kinds of spatial distributions of the qubits. For these, we used the following functions:
 The optimal spatial distribution for the positioning of the qubits is marked in black in Fig.~\ref{fig:max_qfi_vs_k} and is given by
\begin{align}
x_i=\begin{cases}
0 & \mathrm{for}\,\,\,i\le \left\lfloor N/2 \right\rfloor\,,\\
L & \mathrm{for}\,\,\,i > \left\lfloor N/2 \right\rfloor\,.
\end{cases}
\end{align}
The spatial distribution marked by the color darker grey (purple), is given by
\begin{equation}
x_i=\frac{L}{2}\left\{1+ \tanh\left[\left(\frac{2i}{L}-1\right) \pi\right]\right\}\,.
\end{equation}
The equidistant spatial distribution marked in grey (red) is given by
\begin{equation}
x_i=(i-1)\frac{L}{N-1}
\end{equation}
and the spatial distribution marked by lighter grey (yellow), is given by
\begin{equation}
x_i=\frac{L}{2}\left\{1+ \tan\left[\left(\frac{2i}{L}-1\right) \frac{\pi}{4}\right]\right\}\,.
\end{equation}

\end{appendix}

\end{document}